\documentclass[sigconf]{aamas} 

\usepackage{balance} 

\setcopyright{ifaamas}

\usepackage{algorithm}
\usepackage{algorithmic}
\usepackage{graphicx} 
\usepackage{soul,xcolor}

\newif\ifcomments
\commentstrue
\ifcomments
\newcommand{\pallavi}[1]{{\textcolor{red}{Pallavi: { #1}}}}
\else
\newcommand{\pallavi}[1]{}
\fi

\usepackage{amsmath}
\usepackage{amsthm}
\usepackage{xspace}
\usepackage{todonotes}
\usepackage{enumitem}
\usepackage{nicefrac}
\usepackage{multirow}

\usepackage{comment}
\excludecomment{dontshow}

\usepackage{cleveref}

\newtheorem*{theorem*}{Theorem}

\newtheorem{theorem}{Theorem}
\newtheorem{proposition}{Proposition}

\newtheorem{lemma}{Lemma}
\newtheorem{claim}{Claim}
\newtheorem{obs}{Observation}
\newtheorem{remark}{Remark}
\newtheorem{cor}[theorem]{Corollary}
\newtheorem{definition}{Definition}
\theoremstyle{remark}

\newcommand{\FG}{{\sc Problem}\xspace} 
\newcommand{\pc}{\ensuremath{\mathrm{PC}}}
\newcommand{\cpc}{\ensuremath{\mathrm{C}}}
\newcommand{\m}{\ensuremath{\mathcal{M}}}

\newcommand{\bt}{\beta} %
\newcommand{\wh}{\widehat} %
\newcommand{\il}[1]{\todo[inline, color=yellow!50]{[Sushmita] #1}}
\newcommand{\td}[1]{\todo[inline, color=pink!50]{{\bf TO DO:} #1}}
\newcommand{\conf}{\ensuremath{{\tt co}}}

\newcommand{\ma}[1]{}

\newcommand{\Ma}[1]{{\textcolor{magenta}{#1}}}
\newcommand{\nc}[1]{{\textcolor{black}{#1}}}

\newcommand{\hm}{{\tt Ham}} 
\newcommand{\rep}{\mathcal{R}} 
\newcommand{\poly}[1]{\ensuremath{{\rm poly}(#1)}}

\newcommand{\hide}[1]{}

\newcommand{\app}{($\clubsuit$)}

\newcommand{\no}{{\rm no}}

\newcommand{\yes}{{\rm yes}}

\newcommand{\utility}{{utility}\xspace}

\newcommand{\tw}{\mathrm{tw}}
\newcommand{\pw}{\mathrm{pw}}
\newcommand{\FPT}{\ensuremath{{\sf FPT}}\xspace}
\newcommand{\XP}{\ensuremath{{\sf XP}}\xspace}
\newcommand{\NP}{\ensuremath{{\sf NP}}\xspace}

\newcommand{\NPH}{{\sf NP}-{\sf hard}\xspace}
\newcommand{\WOH}{\ensuremath{{\sf W[1]-hard}}\xspace}

\newcommand{\FPTAS}{\ensuremath{{\sf FPTAS}}\xspace}
\newcommand{\AS}{\ensuremath{{\sf AS}}\xspace}

\usepackage{mathtools}
\DeclarePairedDelimiter\ceil{\lceil}{\rceil}

\usepackage{breqn}
\usepackage{amsmath}

\newcommand{\BigO}[1]{\ensuremath{\mathcal{O}(#1)}}
\newcommand{\Co}[1]{\ensuremath{\mathcal{#1}}}
\newcommand{\sm}{\setminus}
\newcommand{\sse}{\subseteq}

\newcommand{\pname}{{\sc BCFEA}\xspace}
\newcommand{\gname}{conflict graph\xspace}
\newcommand{\instance}{$\mathcal{I} = (G,k, p, c,P, B)$\xspace}

\newcommand{\instancetwo}{$\mathcal{I} = (G,2, p, c,P, B)$\xspace}

\newcommand{\type}{\ensuremath{\lambda}}
\newcommand{\size}{\ensuremath{s}}

\newcommand{\items}{items\xspace}

\newcommand{\profit}{\ensuremath{\tau}}
\newcommand{\cost}{\ensuremath{\kappa}}

\newcommand\restr[2]{{
  \left.\kern-\nulldelimiterspace 
  #1 
  \vphantom{\big|} 
  \right|_{#2} 
  }}

\acmConference[AAMAS '24]{Proc.\@ of the 23rd International Conference
on Autonomous Agents and Multiagent Systems (AAMAS 2024)}{May 6 -- 10, 2024}
{Auckland, New Zealand}{N.~Alechina, V.~Dignum, M.~Dastani, J.S.~Sichman (eds.)}
\copyrightyear{2024}
\acmYear{2024}
\acmDOI{}
\acmPrice{}
\acmISBN{}

\acmSubmissionID{1230}

\title[Budget-feasible Egalitarian Allocation of Conflicting Jobs]{Budget-feasible Egalitarian Allocation of Conflicting Jobs}

\author{Sushmita Gupta}
\affiliation{
  \institution{The Institute of Mathematical Sciences}
  \city{Chennai}
  \country{India}
  }
\email{sushmitagupta@imsc.res.in}

\author{Pallavi Jain}
\affiliation{
  \institution{Indian Institute of Technology Jodhpur}
  \city{Jodhpur}
  \country{India}
  }
\email{pallav@iitj.ac.in}

\author{A. Mohanapriya}
\affiliation{
  \institution{The Institute of Mathematical Sciences}
  \city{Chennai}
  \country{India}
  }
\email{mohana@imsc.res.in}

\author{Vikash Tripathi}
\affiliation{
  \institution{The Institute of Mathematical Sciences}
  \city{Chennai}
  \country{India}
  }
\email{vikasht@imsc.res.in}

\begin{abstract}

Allocating conflicting jobs among individuals while respecting a budget constraint for each individual is an optimization problem that arises in various real-world scenarios. In this paper, we consider the situation where each individual derives some satisfaction from each job. We focus on finding a feasible allocation of conflicting jobs that maximize egalitarian cost, i.e. the satisfaction of the \nc{individual who is worst-off}. To the best of our knowledge, this is the first paper to combine egalitarianism, budget-feasibility, and conflict-freeness in allocations.  We provide a systematic study of the computational complexity of finding budget-feasible conflict-free egalitarian allocation and show that our problem generalizes a large number of classical optimization problems. Therefore, unsurprisingly, our problem is \NPH even for two individuals and when there is no conflict between any jobs. We show that the problem admits algorithms when studied in the realm of approximation algorithms and parameterized algorithms with a host of natural parameters that match and in some cases improve upon the running time of known algorithms. 

\end{abstract}


\newcommand{\BibTeX}{\rm B\kern-.05em{\sc i\kern-.025em b}\kern-.08em\TeX}

\begin{document}


\pagestyle{fancy}
\fancyhead{}


\maketitle 


\section{Introduction}
\label{sec:intro}


The division of resources among several interested parties that is satisfactory to all is a central question in game theory and is studied under the name of {\it fair division}. It is among the most commonly encountered challenges in life and industry, such as splitting an inheritance, sharing rent, a partnership dissolution, sharing food, assigning jobs to people or machines, and on and on. Unsurprisingly, research on this forms the basis of a large body of research that spans mathematics, economics, computer science, operations research, and social sciences. The applied nature of the problem notwithstanding, research on fair division can be dated back to mathematics, in the 1948 work of \cite{Steinhaus48}. 


\hide{
The problem is fascinating due to the fact that no common rules serve every purpose, due to differences among resources, preferences of individuals, and situations. This reality is reflected in the myriad ways in which the problem can be modeled, capturing the diversity among the nature of goods--divisible or indivisible; the nature of preferences--cardinal or ordinal; the nature of solution--based on envy-freeness or some objective function such as social welfare, Nash welfare, maximin, leximin, and several others that are too numerous to enumerate; and even the algorithmic paradigm--offline or online. Amanatidis et al. \cite{A..} have recently authored an extensive survey on fair division that }



In recent years, the theoretical AI community has delved into this question from multiple fronts, combining realistic constraints such as allocating conflicting resources that respect budget constraints. Chiarelli et al.~\cite{Chiarelli2022fair}, Biswas et al.~\cite{Biswas_FORC_2023}, and Hummel et al.~\cite{HummelHetland} have studied {\it conflict-free allocation}, one where the bundle of items assigned to any agent can contain at most one of the conflicting items. Such constraints arise quite naturally in scheduling problems, be it for jobs to machines, room assignments to sessions in a conference, or panel selections while avoiding conflicts of interest.

Each of the aforementioned papers study {\it conflict-free allocations} with different objective functions such as partial egalitarian (maximin) allocations~\cite{Chiarelli2022fair}, i.e allocations that maximize the utility of the worst-off agent (also called the Santa Claus guarantee), but where some items may remain unallocated; whereas \cite{HummelHetland} studies complete allocations with fairness criterion such as envy-freeness, maximin share guarantees, and Nash welfare; and \cite{Biswas_FORC_2023} compliments the work of Chiarelli by studying uniform and binary valuations for the setting of course allocation. The other direction in which fair division has been explored is with {\it budget} constraint. Clearly, in most realistic scenarios, a notion of ``budget'' comes into play quite naturally, it can represent money available to purchase items, time available to complete tasks, and so on. Indeed, in recent years papers on the budgeted version of fair division have appeared in top algorithmic research conference venues, both in theory, Garg et al.~\cite{Garg18-budget}, and in AI, Barman et al.~\cite{Barman23-budget}, where the objectives are Nash social welfare and envy-freeness. These papers look at the classical computational complexity divide and devise algorithms for restricted settings of utility function families (such as \cite{Biswas_FORC_2023,HummelHetland}) and for special graph classes (\cite{Chiarelli2022fair}) that capture the conflict among items. The precise definitions of the problem may vary in that \cite{Biswas_FORC_2023} studies the course allocation problem, which is a many-to-many allocation scenario. In some cases, they give approximation algorithms that work in polynomial time. 


In this paper, we extend this line of research by combining these two perspectives and study conflict-free allocation that respects budget constraints such that the worst-off agent attains a given level of utility/satisfaction, formally defined below. For example, consider the challenge facing a CEO of a startup who has a small number of employees, say $k$, and a large number of tasks, say $n$, that she needs to assign to her employees in a manner that respects logistical constraints so that no one is assigned more than one task in any given time, and also give each employee a level of satisfaction while ensuring that nobody's workload is excessive. This is a scenario in which our model, formally defined below, can be employed.


\medskip
\noindent\fbox{
\parbox{0.45\textwidth}{
{\sc Budgeted Conflict Free Egalitarian Allocation} (\pname)\\
{\bf Input:} A set of agents $[k]=\{1,\ldots, k\}$, a set of $n$ \items $V$, each agent's utility function $\{p_i:2^{V} \to \mathbb{Z}_{\geq 0}\}_{i\in [k]}$, and cost function $\{c_i:2^{V} \to \mathbb{Z}_{\geq 0}\}_{i\in [k]}$; a graph on the item set $G=(V, E)$, \nc{called the {\it conflict graph}}; and two positive integers $P$ and $B$. 

{\bf Question:} Does there exists a partition $\mathcal{S} = \{S_1, \ldots S_{k}\}$ of the item set $I$, called {\it bundles}, such that for each $i \in [k]$, $S_i$ is an independent set in $G$, bundle's cost $c_i(S_i)  \le B$, and utility $p_i(S_i) \ge P$?\\
\hide{\medskip An instance of \pname is denoted by $(G,k, \{p_i,c_i\}_{i\in [k]},P, B)$.}
}}

Moreover, in this paper we look at the fair allocation problem from the perspective of parameterized complexity that incorporates an upper bound on budget(or size) (such as \cite{Biswas_FORC_2023}) on (conflict) graph classes that generalize well-known tractable cases (such as \cite{Chiarelli2022fair}) that go beyond maximum degree of the conflict graph (such as studied by \cite{HummelHetland}). Our problem generalizes several classical optimization problems, such as {\sc Partition}, $k$-{\sc Coloring}, {\sc 3-Parition}, {\sc Santa Claus}, {\sc  Bin Packing}, and {\sc Knapsack}. Hence, unsurprisingly, our problem is \NPH even under strict restrictions and thereby sets us well to explore a host of input and output parameters individually and in combination.

Specifically, we point to Chiarelli et al.~\cite{Chiarelli2022fair}, whose work is closely related to ours, as considering various graph classes and exploring the classical complexity boundary between strong {\sf NP}-hardness and pseudo-polynomial tractability for a constant number of agents. Our analysis probes beyond the {\sf NP}-hardness of \pname and explores this world from the lens of parameterized complexity and approximation algorithm, thereby drawing out the suitability of natural parameters--such as the number of agents $k$, the number of \items $n$, the maximum size of each allocated set $s$, the number of distinct profit and cost values, \type. In addition to this, we explore the effectiveness of structural parameters of the underlying graph (such as the number of (connected) components, the treewidth, the pathwidth, and chordality)--towards yielding polynomial time algorithms when the parameters take on constant values. \nc{This is aligned with the emerging area of research in the theoretical AI community that has studied treewidth as a  parameter along with various objectives: connected fair division, Deligkas et al.~\cite{deligkas2021parameterized}; compact fair division, Madathil~\cite{madathil2023fair}; gerrymandering on planar graphs, Dippel et al. ~\cite{dippel2023gerrymandering}. }

\pname is quite recognizable in its restricted avatar as several well-known problems such as $k$-{\sc Coloring}, {\sc Partition}, {\sc Knapsack}, {\sc Bin packing}, {\sc Job Scheduling}~\cite{garey1979computers}, and {\sc Santa Claus}~\cite{Bansal06}. Their connection to \pname can be easily established if for each agent $i\in [k]$, \nc{we take the utility and cost functions to be additive over each item}. That is, we have functions $p_i, c_i:V \to \mathbb{Z}_{\geq 0}$ such that on any subset $S$ of items we define $f_i(S)=\sum_{v\in S} f_i(v)$ for both $f_i\in \{p_i, c_i\}$. We briefly discuss the connection of these problems with \pname. The notation $[k]$ is used to denote the set $\{1,\ldots,k\}$.

\begin{description}[wide=0pt]
\item[$k$-{\sc Coloring}:]Let $(H,k)$ be an instance of $k$-{\sc Coloring}. The goal is to decide if there exists a proper coloring of $H$ using at most $k$ colors. In the reduced instance of \pname, $G=H$; $p_i(v)=1$ and $c_i(v)=0$,  for each $i\in [k]$ and $v\in V(G)$; $P=0$; and $B=0$. 


\smallskip

\item[{\sc Partition}:] Let $S=\{s_1, \ldots, s_n\} \sse \mathbb{Z}$ define an instance of {\sc Partition.} The goal is to find a subset $X \sse S$ such that $\sum_{s_i\in X} s_i = \sum_{s_i \in S \sm X}s_i$. In the reduced instance of \pname, $G$ is edgeless; $k=2$; $p_i(v)=s_i$ and $c_i(v)=0$ for each $i \in [k]$ and $v\in V$; $P=\sum_{v \in V} p_i(v)/2$ and $B=0$. 

\smallskip

\item[3-{\sc Partition}:] The input consists of set $S$ of $3m$ elements, a bound $X \in \mathbb{Z}_{+}$, and a size $s(x) \in \mathbb{Z}_+$  for each $x\in X$ such that $\nicefrac{X}{4} < s(x) < \nicefrac{X}{2}$ and $\sum_{x\in S} s(x)=mB$. The goal is to decide whether there exists a partition of $S$ into $m$ disjoint sets $S_{1}, S_{2}, \ldots, S_{m}$ such that for each $1\leq i \leq m$, $\sum_{x\in S_i} s(x) = B$. In our reduced instance, we set $G$ to be edgeless, $k=m$, $p_i(v)=s_i$, and $c_i(v) = 0$ for each $i \in [k]$ and $v\in V$, $P=X$, and $B=0$.

\item[{\sc Egalitarian Fair Division}\footnote{Also called {\sc Santa Claus}}:]
Let $\Co{I}=(S,k, \{p_i\}_{i\in [k]}, \tau)$ denote an instance of {\sc Santa Claus} problem,~\cite{Bansal06}. The goal is to decide if there exists a partition $(S_1, \ldots, S_k)$ of the set $S$ such that $\min_{i\in k}\{ \sum_{s \in S_i} s\} \geq \tau$. In the reduced instance, $G$ is edgeless, $V=S$,  $c_i(v)=0$ and $p_i(v)=p_i(v)$ for each $i\in [k]$ and $v\in S$ and $B=0$;\hide{\footnote{This objective function is known by various names, Egalitarian and Santa Claus being popular in computational social choice and optimization research, respectively.}}  
\smallskip

\item[{\sc Bin Packing}:] Let $\Co{I}=(X=[n],W,k)$ be an instance of {\sc Bin Packing} where each item $i \in X$ has size $s_i$. The goal is to decide if there exists a partition of $X$ into $k$ parts $(X_1, \ldots, X_k)$ such that for each $j\in [k]$, $\sum_{i \in X_j} s_i \leq B$. In the reduced instance, $G$ is edgeless, and $p_i(v)=0$, $c_i(v)=s_i$ for each $i\in [k]$ and $v\in V$, $B=W$, and $P=0$\hide{i.e we only have cost (or size)}. 
\smallskip

\item[{\sc Knapsack}:] Let $\Co{I}=(V', \{p,w\}_{v\in V}, P',W)$ be an instance of {\sc Knapsack}, \nc{where $p$ and $w$ are the profit and cost functions, respectively, on the item set $V$, and $P',W \in \mathbb{N}$.} The goal is to find a subset of $V$ such that the sum of the profit values and the cost values of the item is at least $P'$ and at most $W$, respectively. In the reduced instance of \pname, we have $V= V' \cup \{d\}$, $G$ is edgeless, $k=2$, and $p_1(v)=p_v$, $p_2(v)=P'$, $c _{1}(v) = w_v$, $c_{2}(v) = 0$ for each item $v\in V$, and for the ``dummy'' item $p_1(d)=0$, $p_2(d)=P'$, $c_1(d)=B+1$, $c_2(d)=0$; and $P=P'$ and $B = W$.


\end{description}

We summarize the hardness results obtained from the above reductions as follows. 
\begin{remark}\label{rem:hardness}
Under ${\sf P} \neq \NP$ and $\FPT \neq \WOH$, we can infer the following.
\begin{enumerate}[label=(\alph*),wide=4pt]
    \item\label{first} Due to the reduction from {\sc Partition}, \pname is \NPH even for $k=2$, there are no conflicts, and $B=0$. Thus, we cannot hope for an $\FPT(k+\tw+B)$ time algorithm.\footnote{For a parameter $\zeta$, $\FPT(\zeta)$ refers to running time of an algorithm that runs in $f(\zeta)\cdot m^{\BigO{1}}$, for any computable function~$f$.}
    \item\label{second} Due to the reduction from $3$-{\sc Coloring}, \pname is \NPH even for $k=3,P=0,B=0$, unit profit function and zero cost function. Thus, we cannot hope for an $\FPT(k+P+B)$ algorithm even when values are encoded in unary;
    \item\label{third} Due to reduction from $3$-{\sc Partition}, \pname is \NPH even when $s=3, B=0$, and graph is edgeless. Thus, we cannot hope for an $\FPT(s +B)$ or $\FPT(\tw+B)$ even when values are encoded in unary.

    \end{enumerate}
Here, $\tw$ and $s$ denote the treewidth of the \gname and the maximum size of a bundle, respectively.
\end{remark}





Thus, the problem {\sc Fair k-division of Indivisible Items} studied by Chiarelli et al.~\cite{Chiarelli2022fair} can be reduced to \pname if for every agent $i\in[k]$ and item $v\in V$,

$c_i(v)=1$ for every agent $i\in[k]$ and item $v\in V$, and the budget $B=n$, while the profit of each item $p_i(v)$

\paragraph{Our Contributions}

Due to ease of exposition we will present the proofs for the case of {\it identical valuation and cost functions}, as defined formally below.
  


\noindent\fbox{
\parbox{0.45\textwidth}{
{\bf Input:} An instance of \pname where for each agent $i\in [k]$, we have identical \utility and cost functions, $p_i=p$ and $c_i=c$. \\
{\bf Output:} A solution for \pname.

\smallskip
An instance of \pname is denoted by \instance.
}}

\smallskip
However, as we discuss in \Cref{ss:generalizing-to-gen-valuations}, each of our algorithms can be easily extended to the general setting of non-identical \utility and cost functions. For simplicity of exposition, throughout the paper, we will use $c_v$ and $p_v$ to denote the cost $c(v)$ and the \utility $p(v)$, respectively, for $v \in V$.


Our work in this paper is a deep dive into the combinatorial and optimization aspects of \pname whereby we show a wide array of algorithmic results that explore the \NPH problem with respect to various parameters. As explained earlier, the formal definition of the problem can be seen as model that captures various well-known and well-studied combinatorial optimization problems. In this sections, we will discuss the relevance of the results presented in this paper vis-a-vis the known literature about the underlying problems. Here, we will discuss the wider backdrop of this work.




\begin{remark}
\label{remark:one_agent}
  When there is exactly one agent, \pname is rather simple to solve as the \gname must be edgeless and the sum of the costs and utilities of all the objects must be at most $B$ and at least $P$, respectively. All these conditions are simple to test in linear time.  
\end{remark}

Next, we discuss the results for the other values of $k$.  


\begin{description}[wide=0pt]

\item[{\bf Number of agents is two:}]\pname exhibits a dichotomy, whereby if the \gname is edgeless, the problem is \NPH and it is polynomial-time solvable if the graph is connected, \Cref{t:poly_2fpcg}. \nc{More generally, we show that \pname is \NPH when the graph is disconnected, Theorem~\ref{thm:disconnected-cg}. However, we supplement this by exhibiting that \pname has an \FPT algorithm with respect to the number of components, \Cref{thm:fpt_cc}.} In contrast to \Cref{rem:hardness}, we note that when $k=2$, \pname can be solved in polynomial-time when the values are encoded in unary, \Cref{thm:k_is_2_graph_arbitrary}.

\nc{We note that Chiarelli et al.~\cite{Chiarelli2022fair} study a restricted version of \pname, where there is no budget constraint, that is, there is no cost associated with any bundle. Our algorithm for the $k=2$ case, 
when applied to their setting (Theorem~\ref{thm:k_is_2_graph_arbitrary}) has a better running time $\BigO{n P^{2}}$, where $P$ is the utility of the worst-off agent, than the one Chiarelli et al. propose in ~\cite{Chiarelli2022fair} (Lemma~$13$), which is $\BigO{n Q^4}$, where $Q$ denotes the sum of the profits of all the items (clearly, $Q \ge P$).} 


\hide{We know that {\sc Bin Packing} is \NPH even when $k=2$. Thus, we cannot even expect an \XP algorithm (i.e it runs in time $n^{f(k)}$) let alone an \FPT algorithm (runs in time $f(k)n^{\BigO{1}}$. However, there is a pseudo-polynomial time algorithm (i.e  when size of the items are polynomially bounded integers, then the problem can be solved in polynomial time).} 

\hide{Moreover, Jansen et al.~\shortcite{Marx13-BinPacking} show that even {\sc Unary Bin Packing} (where values are encoded in unary) is \WOH with respect to $k$. \Ma{Thus, \pname cannot have \FPT algorithm with respect to $k$, even when the \gname is edgeless and the utility and cost values are encoded in unary.} \pallavi{as discussed the relation between two problems earlier}}





\item[{\bf Number of agents is arbitrary:}] 


We present an exact-exponential algorithm, \Cref{thm:exact-algorithm}, with running time $\Co{O}^{\star}(2^{n})$\footnote{$\Co{O}^{\star}(\cdot)$ suppresses polynomial or poly-logarithmic factors.} for \pname, where $n$ denotes the number of items. This is asymptotically tight given that under the Exponential Time Hypothesis (ETH) we know that $k$-{\sc Coloring} cannot be solved in $2^{o(n)}$ time. 


As a consequence of the hardness of \pname with respect to $k$, it is natural for us to probe the parameterized complexity of \pname with respect to $k$ with {\it domain restrictions}, which in the context of our study means that we have restrictions on the graph, the number of pairs of \utility and cost values (formalized as the {\it type}), the range of \utility and cost values (or the size of the encoding) and the {\it bundle size}, defined to be the maximum size of a set assigned to an agent in a solution of \pname. 




{\it Each of the following results assumes that the values are encoded in unary.}
 \begin{itemize}[wide=4pt]
     \item \pname has an $\FPT(k+\tw)$ algorithm, \Cref{thm:ppoly-tw}. As discussed in the technical section, if the \gname is a chordal graph, then any \yes-instance of \pname has treewidth is at most $k$. Specifically, if the \gname is an interval graph, then the pathwidth is at most $k$. Hence, when $k$ is a constant, \pname admits a polynomial-time algorithm when $G$ is an interval graph,~\Cref{cor:chordal}.
    Interval graphs are of practical importance due to their relevance in {\sc Job Scheduling}.



\item{\bf (Type)} \pname admits an $\FPT(k+\tw+\type)$ algorithm,~\Cref{thm:const_type}, where $\type$ denotes the type. Consequently, when both $k$ and $\type$ are constants, the problem is polynomial-time solvable,~\Cref{cor:constant_type}. Due to \Cref{rem:hardness}, we cannot expect an $\FPT(k+\tw)$ algorithm or $\FPT(k+\type)$ algorithm.

\item{\bf (Bundle size)} Due to    \Cref{rem:hardness}\ref{second} we cannot expect an $\FPT(k+P+B)$ algorithm. Hence, we look at $s$, the {\it bundle size}, as a parameter. The bundle size is likely to be small, but again due to \Cref{rem:hardness}\ref{third} we cannot expect an $\FPT(s)$ algorithm either. But for the case $s=2$, we have a polynomial-time algorithm,~\Cref{thm:bundlesize-2}. More generally, we show that we can have an $\FPT(k+s)$ algorithm, \Cref{thm:agent-and-bundle-size}.

     \item {\bf (FPT-AS)} \pname has an $\FPT(k+\tw)$ algorithm that outputs a solution in which every agent has profit at least $P/(1+\epsilon)$ and cost at most $(1+\omega)B$ (such algorithms are known as \FPT-{\it approximation schemes}), \Cref{thm:FPTAS-tw}.






     \end{itemize}


\noindent{\bf Choice of graph classes:} As discussed above, \pname generalizes $k$-{\sc Coloring}, it is the only graph problem among the ones discussed above. Thus, any hope of designing an algorithm for \pname rests on the tractability of $k$-{\sc Coloring} in the \gname. Hence, our search for amenable graph classes is narrowed down to interval graphs, chordal graphs, and graphs with bounded treewidth. In chordal graphs, it is polynomial-time solvable, and in graphs of bounded treewidth it is $\FPT(\tw)$,~\cite{pc_book}. Interval graphs are of special significance to applications related to {\sc Job Scheduling}. Our work in this article runs the gamut of designing algorithms for appropriate small valued parameters and graph classes where the underlying combinatorial problem is tractable. We bookend our algorithmic results with hardness borrowed from {\sc Bin Packing} in addition to $k$-{\sc Coloring}.
\end{description}

Please refer to Table~\ref{table:results} for an overview of the results in this paper. In the technical section, we have explained the key ideas in each result. 

\begin{table*}[ht!]
\begin{center}
\begin{tabular}{|c|c|c|c|}
\hline
\#Agents & Graph &Complexity & Reference  \\
\hline
1 & General & \BigO{n} & \Cref{remark:one_agent} \\
\hline
\multirow{4}{*}{} & Connected & \BigO{n} & \Cref{t:poly_2fpcg}  \\ 
2 & $r$  components, $r>1$  & \NPH & \Cref{thm:disconnected-cg}\\
& $r$ components  & \BigO{2^r n} & \Cref{thm:fpt_cc}   \\
 & General & \BigO{(BP)^{2} n} & \Cref{thm:k_is_2_graph_arbitrary}  \\
\hline
\multirow{5}{*}{}  & Bounded Treewidth  & \BigO{\tw^{k} (\alpha \gamma)^{2k}n }& \Cref{thm:ppoly-tw}  \\
\BigO{1} & Interval Graph & \BigO{k^{k} (\alpha \gamma)^{k} n } & \Cref{cor:chordal}   \\
& Bounded Treewidth & $\BigO{\tw^{k} (\log_{1+\epsilon} \alpha  \log_{1+\omega} \gamma )^{2k} n^{\BigO{1}} }$ & \Cref{thm:FPTAS-tw}\\
 & (Bi-criteria approximation) & $\left( \frac{1}{1+\epsilon},1+\omega \right)$-factor approx& \\
\hline
\multirow{4}{*}{}  & General & $\Co{O}^{\star}(2^{n})$ & \Cref{thm:exact-algorithm}   \\
Arbitrary & General & $\FPT(k+s)$  & \Cref{thm:agent-and-bundle-size} \\
 & Bounded Treewidth & $\BigO{\tw^k (\type+1)^{2k}  k  n}$  & \Cref{thm:const_type} \\
\hline
\end{tabular}
\end{center}
\caption{Overview of our algorithms. Here, $\gamma=\sum_{v\in V}c(v)$, $\alpha=\sum_{v\in V}p(v)$, $\lambda = |\{(c_v,p_v)\mid \forall v \in V\}|$, $\tw$ denote the treewidth of $G$ in an instance \instance of \pname, $s$ denote the size of bundle, and $\epsilon, \omega \in (0,1]$. }
\label{table:results}
\end{table*}


\newcommand{\cffa}{{\sc CFFA}\xspace}

\paragraph{Related Work.} In addition to the works discussed earlier in the Introduction, we will further discuss some more work that are related to our work in this paper.
A well-known framework within which fair allocation has long been studied is in the world of {\sc Job Scheduling} problems on non-identical machines. In this scenario, the machines are acting as agents and the jobs are the tasks such that certain machines are better suited for some jobs than others and this variation is captured by the ``satisfaction level" of the machine towards the assigned jobs. Moreover, the jobs have specific time intervals within which they have to be performed and only one job can be scheduled on a machine at a time. \hide{Thus, the subset of jobs assigned to a single machine must respect these constraints, and the objective can be both maximization and minimization as well as just feasibility.} Results on the computational aspect of fair division that incorporates interactions and dependencies between the items are relatively few. This is the backdrop of our work in this article.  A rather inexhaustive but representative list of papers that take a combinatorial approach in analysing a fair division problem and are aligned with our work in this paper is \cite{cheng2022restricted,bezakova2005allocating,kurokawa2018fair,DBLP:conf/atal/EbadianP022,DBLP:conf/atal/BarmanV21,ahmadian2021four,Chiarelli2022fair,Warut21,woeginger1997polynomial}. In particular, we can point to the decades old work of Deuermeyer et. al~\cite{deuermeyer1982scheduling} that studied a variant of {\sc Job Scheduling} in which they goal is to assign a set of independent jobs to identical machines in order to maximize the minimal completion time of the jobs. Their {\sf NP}-hardness result for two machines (i.e. two agents in our setting) is an early work with a similar flavor. They analyse a well-known heuristic called the LPT-algoirthm to capture best-case performance and show that its worst case performance is 4/3-factor removed from optimum.



Moreover, we note that conflict like constraints in an underlying graph have also been studied in the context of {\sc Knapsack}~\cite{knapsack-conflict09,knapsack-conflict17} and {\sc Bin Packing}~\cite{BinPacking-conflict} have also been studied. Interestingly, Pferschy and Schaue~\cite{knapsack-conflict09} studies {\sc Knapsack} with conflict and present pseudo-polynomial algorithms for graphs of bounded treewidth and chordal graphs. From these algorithms, they derive fully polynomial time approximation schemes (\FPTAS). More generally, there can be other combinatorial notions of quality that capture scenarios where certain subsets of items, based on what they represent, should be bundled together and thus be assigned to the same agent. Conflict is a special case of this in which we want independence between the items, connectivity is another well-studied notion. Such combinatorial properties are amenable to graph-theoretic modeling and have been studied as well~\cite{Bouveret17,DARMANN20,Guy09,Factorovich20}. Suksompong~\cite{Warut21} surveys the landscape of fair division with a host of combinatorial constraints such as connectivity, cardinality, matroid, geometric, separation, budget, and conflict.  

\hide{Some papers that have used conflict graphs to capture various constraints on items/jobs that are related to compatibility are~\cite{DARMANN20, Guy09,Factorovich20,Bouveret17}. Results pertaining to approximation algorithms exist, and Bansal et al.\cite{Bansal06} studied the {\sc Santa Claus} problem, which corresponds to the variant of the problem studied by Chiarelly et al. when the number of agents is part of the input. Since then various approximation results have appeared exploring different concepts of objective functions and various approximation measures, see, e.g., \cite{Asadpour10,Chakrabarty09}.}

Other paradigmatic problems that are subsumed by \pname is {\sc Bin Packing} and {\sc Knapsack} and the literature around them--both offline and online--is so immense that we cannot faithfully survey it here and point the reader to \cite{BinPacking-survey} and \cite{Knapsack-book} for a deeper look. 


We would like to conclude this discussion by noting that the literature on fair division is {\it vast} due to the myriad of variations among the nature of goods--divisible or indivisible; the nature of preferences--cardinal or ordinal; the nature of solution--based on envy-freeness or some objective function such as social welfare, Nash welfare, maximin, leximin, and several others that are too numerous to enumerate; and even the algorithmic paradigm--offline or online. Amanatidis et al.~\cite{fair_allocation-combined} have recently authored an extensive survey on fair division that describes the research on the topic from the perspective of theoretical computer science.







\section{Preliminaries}
\noindent {\em Graph Theoretic Notations:} For the graph theoretic notations, not defined here, we refer to Diestel's~\cite{Diestel-book}.  We use the standard notation $[n]$ to denote the set $\{1,2, \ldots, n\}$. Let $G$ be a graph with vertex set $V(G)$ and edge set $E(G)$. The set $N_{G}(v) =\{u \in V \mid uv \in E(G)\}$ is called the \emph{open neighbourhood} and $N_{G}[v] = N_{G}(v) \cup \{v\}$ is called the \emph{closed neighbourhood} of $v$ in $G$. The {\em degree} of a vertex $v$ is the cardinality of $N_{G}(v)$, denoted by $d_{G}(v)$. A vertex of degree $0$ is called an isolated vertex. For a set $S \subseteq V$, $G[S]$ denotes the subgraph of $G$, induced by set $S$, whose vertex is $S$ and edge set is $\{xy \in E(G) \mid x, y \in S \}$.  A subset $S\subseteq V(G)$ is called an {\em independent set} (resp. {\it clique}) in $G$ if the subgraph induced by $S$ is edgeless (resp. complete).  


\hide{A vertex of degree zero is called an isolated vertex. For a set $S \subseteq V$, $G[S]$ denotes the subgraph of $G$, induced by set $S$, whose vertex is $S$ and edge set is $\{xy \in E(G) \mid x, y \in S \}$.  Let $S \subseteq V(G)$. }


Consider sets $A$, $B$, and $C$ such that $A \subseteq B$. If $f:B \mapsto C$ is a function then the {\em restriction} of $f$ to $A$, denoted as $\restr{f}{A}$, is a function $\restr{f}{A} : A \mapsto C$ defined as $\restr{f}{A}(x) =f(x)$ for all $x \in A$. Let $G$ be a graph. For a subset $S \sse V(G)$, $k$ subsets $X_1,X_2, \ldots, X_k$ of $S$ are called {\em $k$-independent partition} of $V(G)$ if $\bigcup_{i \in [k]}X_{i} = S$ and for each $i \in [k]$, $X_{i}$ is an independent set in $G[S]$. A function $f:V(G) \mapsto [k]$ is called a {\em $k$-coloring} of $G$, if for any edge $uv \in E(G)$, $f(u) \neq f(v)$. Let $A,B \subseteq V(G)$ such that $A \subseteq B$. Let $f$ and $g$ be a $k$-coloring of $G[A]$ and $G[B]$, respectively. We say that coloring $g$ {\em agrees} with coloring $f$ on $A$ if $\restr{g}{A} = f$. More precisely, $g(u) = f(u)$ for all $u \in A$.

For an instance \instance of \pname, we denote the total cost and utility of items to be $\gamma=\sum_{v\in V}c(v)$ of $V$ and $\alpha=\sum_{v\in V}p(v)$, respectively. We say that a partition $(S_1, \ldots, S_k)$ {\it satisfies the cost and profit constraint} $(P,B)$ if for each $i \in [k]$, $c(S_i) \le B$ and $p(S_i) \ge P$. Moreover, if each $S_i$ is an independent set in $G$, then call each $S_i$ a {\it bundle}.  We define the {\it type} of \Co{I} to be $|\{(c_i,p_i)\mid \forall i\in [k]\}|$, denoted by $\lambda$. 
As discussed, {\sc $k$-Coloring} is a special case of \pname. Due to this, we may assume, without loss of generality, that the \gname in an instance of \pname does not have clique of size greater than $k$. 
\begin{obs}\label{obs:bounded-clique-size}
In a given instance of \pname the underlying graph has clique size of at most $k$.
\end{obs}

\medskip
\noindent {\em Tree and Path Decomposition:}

\begin{definition}
  A tree decomposition of a graph $G=(V,E)$ is a pair $(T, \{\bt_t\}_{t \in V(T)})$ where $T$ is tree and $\bt_{t} \subseteq V$, satisfying the following properties:
  \begin{enumerate}
      \item $\bigcup_{t \in V(T)} \bt_{t} =  V$.
      \item for each edge $uv \in E$, there exists a bag that contains both $u$ and $v$.
      \item\label{tw:coherence} Let $T_v = \{t \mid v \in \bt_t\}$. The subgraph induced by the nodes in the set $T_v$ is a subtree of $T$.
  \end{enumerate}
\end{definition}

The {\em width} of a tree decomposition $(T, \{\bt_t\}_{t \in V(T)})$ of graph $G$ is given by $\max \{|\bt_t| -1 \mid t \in V(T) \}$. The minimum width of a tree decomposition of a graph $G$ is called the {\em treewidth} of $G$. \nc{For a node $t \in V(T)$, let $T_{t}$ denotes the subtree of $T$ rooted at $t$ and $V_{t} = \bigcup_{t \in V(T_t)}\bt_t$.}  At each node $t$, we associate a subgraph $G_t$ of $G$ where, $G_{t} = G[V_t]$. The following property of a tree decomposition is very useful while designing algorithms on tree decomposition. 

\begin{proposition}{\rm \cite{pc_book}}
\label{thm:nice_tree_prop}
    Let $(T, \{\bt_t\}_{t \in V(T)})$ be a tree decomposition of a graph $G$ and $t_1t_2$ be an edge of $T$. Further, let $T_{t_1}$ and $T_{t_2}$ be two connected components of $T \setminus \{t_1t_2\}$, $A =  \left(  \bigcup_{t \in V(T_{t_1})} \bt_t \right)  \setminus (\bt_{t_1} \cap \bt_{t_{2}})$, and $B =  \left( \bigcup_{t \in V(T_{t_2})}\bt_t \right) \setminus (\bt_{t_1} \cap \bt_{t_{2}})$. Then no vertex of $A$ is adjacent to a vertex of $B$.
\end{proposition}

For our purpose, we rather use a special tree decomposition of a graph, called {\em nice tree decomposition}. In a nice tree decomposition, we consider $T$ to be a rooted tree with root $r$, and other vertices are identified as a leaf node, vertex introduce node, forget node, or join node. In particular, we consider nice tree decomposition with edge introduce node. In a nice tree decomposition having edge introduce nodes, the following conditions are satisfied:
\begin{itemize}
    \item $\bt_r = \emptyset$ and $\bt_{t} = \emptyset$ for each leaf node $t$ of $T$.
    \item {\em Introduce node}: A node $t$ with exactly one child $t'$ such that \nc{$\bt_{t} = \bt_{t'} \cup \{v\}$} for some vertex $v \notin \bt_{t'}$ ($v$ is introduced at $t$). 
    \item {\em Forget node}: A node $t$ with exactly one child $t'$ such that $\bt_{t} = \bt_{t'} \setminus \{v\}$ for some vertex $v \in \bt_{t'}$ ($v$ is forgotten at $t$). 
    \item {\em Join node}: A node $t$ with exactly two children $t_1$ and $t_2$ such that $\bt_{t} = \bt_{t_1} = \bt_{t_2}$. 
\end{itemize}

We have the following relation between the path decomposition and nice path decomposition of a graph $G= (V,E)$, see~\cite{pc_book} (Chapter $7$, Lemma~$7.4$).

\begin{proposition}\cite{pc_book}
\label{t:nice_tree_time}
    If a graph $G$ admits a tree decomposition of width at most $\tw$, then it also admits a nice tree decomposition of width at most $\tw$. Moreover, given a tree decomposition $(T, \{\bt_t\}_{t \in V(T)})$ of $G$, a nice tree decomposition of $G$ of same width, can be computed in time $O(\tw^{2} \cdot \max(|V(T)|, |V(G)|)$ that has at most $O(\tw \cdot |V(G)|)$ nodes.
\end{proposition}

When $T$ is simply a path then $(T, \{\bt_t\}_{t \in V(T)})$ is called a {\em path decomposition} of the graph $G$. Since $T$ is a path, nodes of $T$ can be ordered as $(t_1, t_2, \ldots, t_r)$ such that $t_{i}t_{i+1} \in E(T)$ for each $i \in [r-1]$. Thus the path decomposition of a graph is also represented as $(\mathcal{P}, (X_1,X_2, \ldots, X_r))$ where $X_i \subseteq V$. The minimum width of a path decomposition of a graph $G$ is called the {\em pathwidth} of $G$. Next, we define the {\em nice path decomposition} of an interval graph.

\begin{definition}
  A path decomposition $(\mathcal{P}, (X_1,X_2, \ldots, X_r))$ of a graph $G=(V,E)$ is called a nice path decomposition if $X_1 =X_r = \emptyset$ and the other nodes of $\mathcal{P}$ is identified as one of the following node: 
  \begin{enumerate}
      \item {\bf introduce node}: a node $t$ is an introduce node if it has exactly one child say $t'$ such that $X_t = X_{t'} \cup \{x\}$ for some vertex $x \notin X_{t'}$.
      \item {\bf forget node}: a node $t$ is an introduce node if it has exactly one child say $t'$ such that $X_t = X_{t'} \setminus \{x\}$ for some vertex $x \in X_{t'}$.
  \end{enumerate}
\end{definition}

\medskip
\noindent{\it Approximation Algorithm.} A fully polynomial-time approximation scheme {(\FPTAS)} is a family of algorithms
$\{A_{\epsilon} \}$, such that for each $\epsilon > 0$, there is a $(1 + \epsilon)$-approximation algorithm $A_{\epsilon}$ with running time bounded by a polynomial in $1/\epsilon$ and the size of input instsubsetsance. An \FPT approximation scheme (\FPT-\AS)\footnote{Not to be confused with a fully polynomial-time approximation scheme (\textup{FPTAS}).} is a \emph{fixed-parameter tractable approximation scheme} is an approximation algorithm that, given an instance $I$ of a maximization problem $\Pi$, an integer $k$, and any $\epsilon > 0$, returns a $\frac{1}{(1+\epsilon)}$-approximate solution in $f(k)\cdot |I|^{\BigO{1}}$ time, where $f$ is an
arbitrary computable function depending on the parameter $k$.

\medskip
\noindent{\it Parameterized complexity.} The goal of parameterized complexity is to find ways of solving \NPH problems more efficiently than exhaustive search: the aim is to
restrict the combinatorial explosion to a parameter that is likely to 
much smaller than the input size in families of input instances. Formally, a {\em parameterization}
of a problem is assigning an integer $k$  to each input instance of the problem.  We
say that a parameterized problem is {\em fixed-parameter tractable
  }(\FPT)  if there is an algorithm that solves the problem in time
$f(k)\cdot |I|^{\BigO{1}}$, where $|I|$ is the size of the input and $f$ is an
arbitrary computable function depending on the parameter $k$
only. A more general class of parameterized algorithm is the XP algorithms where a parameterized problem is {\em slicewise poly
  }(\XP)  if there is an algorithm that solves the problem in time
$|I|^{f(k)}$, where $|I|$ is the size of the input and $f$ is an
arbitrary computable function depending on the parameter $k$
only. Moreover, we will refer to such algorithms as an \FPT  (resp. \XP) algorithm and the problem to have an $\FPT(k)$ (resp. $\XP(k)$) algorithm. For more details on the subject, we refer to the textbooks~\cite{pc_book,fg,downey}.




\section{Our Results} 
\subsection{When the number of agents is two}\label{subsection:k_equals_2}


In this section, we study \pname when $k=2$. We show that \pname exhibits a dichotomy: when the \gname is \nc{disconnected} it is \NPH and it has a polynomial-time algorithm when it is connected. More interestingly, it is \FPT with respect to the number of components. Furthermore, due to \Cref{rem:hardness}, the problem is \NPH for $k=2$ and large profit values, however, in  \Cref{thm:k_is_2_graph_arbitrary}, we design a polynomial time algorithm when values are in unary. \hide{We begin our analysis with a result that shows that \pname is at least as hard as {\sc Partition} when the \gname has no edges. 

\begin{theorem}
\label{t:hard_2fpcg}\app
When $k=2$ and the \gname is edgeless, \pname is {\em NP}-complete. 
\end{theorem}


In Theorem~\ref{t:hard_2fpcg}, while constructing a corresponding instance \instancetwo of \pname, we can set $c_{i}=1$, $u_{i}=s_i$, $B = n$ and $P =\nicefrac{\sum\limits_{i \in [n]} s_{i}}{2}$ and argue similarly to show the hardness of \pname for these variants as well. 

\begin{cor} When $G$ is edgeless and at least one of $B$ or $P$ is $\poly{n}$, we have 
\begin{enumerate}
 \item[(a)] \pname is {\em NP}-complete; \hide{even when $B=\poly{n}$} and 
 \item[(b)] \pname is para NP-hard with respect to $k$. 
\end{enumerate}
\end{cor}}

We begin with the observation that for the case $k=2$ a solution for \pname is a solution for 2-{\sc Coloring} in the \gname. This, yields the following.

\begin{obs}\label{obs:k-is-2} 
When $k=2$, the \gname in a \yes-instance of \pname must be bipartite.
\end{obs}

Thus, for the case $k=2$, the interesting case is when \gname is in fact bipartite.  
\nc{In light of Remark~\ref{rem:hardness}, we know that \pname is already known to be \NPH when the \gname is edgeless. By slightly tweaking the reduction from {\sc Partition}: by adding dummy items that create a matching with the ``real'' items, we can show that the hardness carries forward for the case when the \gname has edges but is still disconnected.}

\begin{theorem}\label{thm:disconnected-cg}
When $k=2$, \pname is \NPH when the conflict graph is disconnected.
\end{theorem}
\begin{proof}
    Clearly, given a two partition $S_1$ and $S_2$ of $V$, in $O(n^{2})$-time we can verify if $S_1$ and $S_2$ are independent. We can also verify if $c(S_i) \le B$ and $p(S_i) \ge P$ for each $i \in [2]$, in polynomial time. Thus the problem is in NP. To show the hardness of the problem, we describe a polynomial reduction from the {\sc Partition } problem which is known to be NP-complete. In the {\sc Partition } problem we are given a set $S = \{s_1,s_2, \ldots , s_n\}$ of integers and the goal is to partition $S$ into two sets $X$ and $Y$ such that $\sum\limits_{s_i \in X} s_i = \sum\limits_{s_j \in Y} s_j $. 
    
    Given an instance  $S = \{s_1,s_2, \ldots, s_n\}$ of the {\sc Partition} we construct an instance of \pname problem as follows: corresponding to each integer $s_i \in S$, we consider two item $u_i,v_i$. Next, we construct the \gname $G=(V,E)$ such that $V= \{u_i,v_i \mid [n]\}$ and $E= \{u_iv_i \mid i \in [n]\}$. We set set $c_{v_i} = s_{i}$ and $p_{v_i} = 1$ for all $i \in [n]$. Moreover, we set, $c_{u_i} = p_{u_i} = 0$ for each $i \in [n]$. Finally, we set $B= \frac{\sum\limits_{i\in [n]}s_i}{2}$ and $P=1$. Now, we prove the following claim:
\begin{claim}
\label{c:np_2}
$S$ can be partitioned  into two sets $X$ and $Y$ such that $\sum\limits_{s_i \in X} s_i = \sum\limits_{s_j \in Y} s_j$ if and only if the corresponding instance \instancetwo of \pname is a \yes~instance.
\end{claim}
\begin{proof}
    Suppose $S$ can be partitioned into two sets $X$ and $Y$ such that $\sum\limits_{s_i \in X} s_i = \sum\limits_{s_j \in Y} s_j$. Note that the sets $S_i = \{v_i \mid s_i \in X\} \cup \{u_j \mid s_j \notin X\}$ and $S_2 = \{v_i \mid s_i \in Y\} \cup \{u_j \mid s_j \notin Y\}$ forms a partition of $V$. Further, $c(S_1) = \sum\limits_{s_i \in X} s_i = \sum\limits_{s_i \in S} s_i/{2} = B$. Similarly, $c(S_2) = \sum\limits_{s_i \in Y} s_i = \sum\limits_{s_i \in S} s_i/2 = B$. Also, $p(S_1), p(S_2) \ge 1$. Thus, $(G,c,p,B,P)$ is a yes-instance of \FG.

    Conversely, suppose $(G,c,p,B,P)$ is a {\sc Yes}~instance of \FG. Then there is partition of $V$ into two sets $S_1$ and $S_2$ satisfying all the conditions. Let $X_{i} = \{s_j \mid v_j \in S_i\}$ for $i \in [2]$. It may be noted $2B = \sum\limits_{v_i \in V}c_{v_i}$ and $c_{v_i} = s_i$ for each $i \in [n]$. This implies that $c(S_1) =  c(S_2) = B$. Thus $\sum\limits_{s_i \in X_i}s_i = \sum\limits_{v_i \in S_i }c_{v_i} = B$ for each $i \in [2]$. This concluded the claim.
\end{proof}

The theorem follows from \Cref{c:np_2}.
\end{proof}


Notwithstanding this hardness, we can show that \pname does admit a parameterized algorithm with respect to the number of components. This allows us to infer that when \gname is connected, it is infact polynomial-time solvable.

\begin{theorem}
\label{thm:fpt_cc}
When $k=2$, \pname admits an $\FPT(r)$ algorithm, where $r$ denotes the number of components in the \gname.
\end{theorem}

\begin{proof}
Consider an instance \instancetwo of \pname. Due to Observation~\ref{obs:k-is-2}, if $G$ is not a bipartite graph then $\mathcal{I}$ is a \no-instance. Here, we assume that $G$ is a bipartite graph.  Let $C_1, C_2, \ldots, C_r$ denote the connected components of $G$. Each $C_i$ is an isolated vertex or a bipartite graph, say, with bipartition $(X_i, Y_i)$. Since $G[C_i]$ is connected, the bipartition of $C_i$ is unique up to renaming the partitions. Let $(S_1, S_2)$ be a partition of $V(G)$ satisfying \utility and cost constraints. Note that if a component $C_i$ is an isolated vertex then either it goes to $S_1$ or to $S_2$; otherwise, since $G[C_i]$ is a bipartite graph with bipartition $(X_i, Y_i)$, all the vertices in $X_i$ will be together in $S_1$ or in $S_2$. Moreover, if vertices of $X_i$ are in the set $S_1$ then vertices of the set $Y_i$ must be in the set $S_2$ and if vertices of $X_i$ are in the set $S_2$ then vertices of the set $Y_i$ must be in the set $S_1$. Thus, in $\BigO{2^{r}\cdot n}$ we can compute a required $2$-independent partition if one exists.
\end{proof}


\smallskip
Therefore, \Cref{thm:fpt_cc} yields the following.


\begin{cor}\label{t:poly_2fpcg}
When $k=2$ and the \gname is connected, \pname admits a polynomial-time algorithm. 
\end{cor}



\hide{The case $k=2$ case exhibits a dichotomy as exhibited by~\Cref{rem:hardness} and the following result. }

\nc{Hence, we can conclude that \pname exhibits a dichotomy with respect to connectivity: it is ``hard'' when the graph is disconnected and ``easy" when it is connected.} {The next result shows that when the representation is in unary, we can solve \pname in polynomial-time. Alternately worded, this means that when $k=2$, \pname is pseudo-polynomial time solvable.} The following result is based on a path-style dynamic programming approach on a {\it layered graph}. We present the key idea and the formal analysis is in the Supplementary.


\begin{theorem}\label{thm:k_is_2_graph_arbitrary}
When $k=2$, \pname admits an algorithm with running time $\BigO{nP^2 B^2}$. That is, it has a polynomial-time algorithm when values are in unary.


\end{theorem}
\begin{proof}
Let \instance be an instance of \pname, where $k=2$. As observed, if $G$ is not a bipartite graph, then $\mathcal{I}$ is a \no-instance. Thus, we assume that $G$ is a bipartite graph. From \Cref{thm:fpt_cc}, we infer that if $G$ is connected, then \pname is polynomial-time solvable. Thus we assume that the graph $G$ is not connected. Let $C_1,C_2, \ldots, C_r$ be component of $G$. Hence, $G[C_i]$ is either an isolated vertex or a connected bipartite graph, denoted by $(X_i,Y_i)$. Let $(S_1,S_2)$ be a solution of the instance \Co{I}. Note that either all vertices of $X_i$ will be together in the set $S_1$ and thus, $Y_i$ will be in $S_2$, or it is the other way around.

Now, for each component $C_i$ we take two vertices $x_i$ and $y_i$. We define \utility on vertices $x_i$ and $y_i$ as $p_{x_i} = \sum_{v \in X_i}p_v$ and $p_{y_i} = \sum_{v \in Y_i}p_v$. Similarly, we define the cost of these vertices as $c_{x_i} = \sum_{v \in X_i}c_v$ and $c_{y_i} = \sum_{v \in Y_i}c_v$. For the component having an isolated vertex $v$, we define $p_{x_i} = p_v$, $c_{x_i} = c_v$, $p_{y_i} = 0$, and $c_{y_i} = 0$. Moreover, $c_s = p_s = 0$.

We define a directed graph $H$ and call it a layered graph with $r$ layers, where layer $i$ consists of vertices $\{x_i, y_i\}$. Each arc in $H$ is either colored red or blue. There are two arcs from $x_{i}$ to $x_{i+1}$; one is colored red, and the other is colored blue, for each $i \in [r-1]$.
Similarly, for each $i \in [r-1]$, there is a red arc as well as a blue arc from $y_{i}$ to $y_{i+1}$, $x_{i}$ to $y_{i+1}$, and $y_{i}$ to $x_{i+1}$.
Additionally, we create a special vertex $s$ and add both red and blue arcs from $s$ to $x_{1}$ and $s$ to $y_{1}$. For a directed path $P$ in $H$, we define \utility of $P$, denoted as $p(P) = \sum_{u\in V(P)} p_{u}$,  and cost of $P$, denoted as $c(P)$, to be $\sum_{u\in V(P)} c_{u}$.

It is easy to check that finding a $2$-independent partition of $V(H)$ is equivalent to finding two $r$-length paths, $P_1$ and $P_2$ such that all arc on $P_1$ are red (call it red path), all arcs on $P_2$ are blue (call it blue path) and $V(P_1) \cap V(P_2) = \{s\}$ ($P_1$ and $P_2$ do not share a vertex other than $s$).
Therefore, to check if $\mathcal{I}$ is a \yes-instance, our objective is to check if there exist $r$-length vertex disjoint red and blue paths starting from $s$ whose total utility is at least $P$ and cost is at most $B$. Here it is important to note that if $P_1$ and $P_2$ are $j$-length red and blue paths starting from $s$ such that they do not share a vertex other than $s$ then the end vertices of $P_1$ and $P_2$ are in set $\{x_j,y_j\}$. Moreover since they do not share a vertex other than $s$, if $P_1$ ends at $x_j$ ($y_j$) then $P_2$ ends at $y_j$ ($x_j$). Now, to achieve our goal, we define a dynamic programming based algorithm. 



For $j\in [r]$, positive integers $\profit_1, \profit_2 \leq P$ and $\cost_1, \cost_2 \leq B$, and $z \in \{1,2\}$ we define a Boolean valued function $\m[j,\profit_1,\cost_1, \profit_2, \cost_2,z]$ such that the function $\m[j,\profit_1,\cost_1, \profit_2, \cost_2,1] = 1$ if there are $j$-length red path and blue path $P_1$ and $P_2$ starting from $s$, satisfying $p(P_1) \ge \profit_1$, $c(P_1) \le \cost_1$, $p(P_2) \ge \profit_2$, $c(P_2) \le \cost_2$ and $P_1$ is ending at~$x_j$.

Similarly, $\m[j,\profit_1,\cost_1, \profit_2, \cost_2,2] = 1$ if there are $j$-length red and blue path $P_1$ and $P_2$ starting from $s$, satisfying $p(P_1) \ge \profit_1$, $c(P_1) \le \cost_1$, $p(P_2) \ge \profit_2$, $c(P_2) \le \cost_2$, and and $P_2$ is ending at $x_j$.

We define the following base cases : 
\begin{itemize}
\item $\m[1,\profit_1,\cost_1, \profit_2, \cost_2, 1]=1$ if $\profit_1 \le p_{x_1}$, $\cost_1 \ge c_{x_1}$, $\profit_2 \le p_{y_1}$, and $\cost_2 \ge c_{y_1}$; and $0$, otherwise. 

\item $\m[1,\profit_1,\cost_1, \profit_2, \cost_2, 2]=1$ if $\profit_1 \le p_{y_1}$, $\cost_1 \ge c_{y_1}$, $\profit_2 \le p_{x_1}$, and $\cost_2 \ge c_{x_1}$; and $0$, otherwise. 
\end{itemize}

Next, we define the following recurrence relation: 

\noindent
\fbox{
\parbox{0.45\textwidth}{
$\m[j,\profit_1,\cost_1, \profit_2, \cost_2,1] = \\ 
\lor_{z \in \{1,2\}} \m[j-1,\profit_1-p_{x_j},\cost_1-c_{x_j}, \profit_2 - p_{y_j}, \cost_2-c_{y_j},z]$,

\smallskip
and\\
$\m[j,\profit_1,\cost_1, \profit_2, \cost_2,2] = \\ 
\lor_{z \in \{1,2\}} \m[j-1,\profit_1-p_{y_j},\cost_1-c_{y_j}, \profit_2 - p_{x_j}, \cost_2-c_{x_j},z]$.
}}


\medskip
From the recursive definition, we infer that the instance $\mathcal{I}$ is a \yes-instance if and only if $\lor_{z \in \{1,2\}} \m[r,P,B,P, B,z]$ evaluates to $1$. Next we show the correctness and analyze the time complexity of the proposed algorithm. We denote a red path by subscript $1$ and a blue path with subscript $2$.

In the following discussions, we will treat paths that only intersect at $s$ to be vertex disjoint. We will say that entry $(j, \profit_1, \kappa_1, \profit_2, \kappa_2, z)$ is {\it computed correctly} if the following conditions holds: \\$\m[j, \profit_1, \kappa_1, \profit_2, \kappa_2, z]=1$ if there is $j$-length red path, denoted by $P_1$, and a $j$-length blue path, denoted by $P_2$ that are vertex disjoint and $u(P_1) \geq \profit_1$, $u(P_2)\geq \profit_2$, and $c(P_1) \leq \cost_1$ and $c(P_2) \leq \cost_2$, and that $P_1$ ends at $x_j$ when $z=1$, and at $y_j$ when $z=2$. We will prove this by induction on the table entry.

The base case is given by $\m[1, p_{x_1},c_{x_1}, p_{x_2}, c_{x_2}, 1]$ and \\$\m[1, p_{x_1},c_{x_1}, p_{x_2}, c_{x_2}, 2]$. 
Suppose that the induction hypothesis holds for all entries up to $(j, \widehat{\profit_1}, \widehat{\kappa_1}, \widehat{\profit_2}, \widehat{\kappa_2}, 1)$. 
We will argue the inductive step for this entry. 
By the recursive definition, we have the value of the above entry given by
\[\lor_{z\in \{1,2\}}\m[j-1, \widehat{\profit_1}-p_{x_{j}}, \widehat{\kappa_1}-\! c_{x_{j}}, \widehat{\profit_2}-\! p_{y_{j}}, \widehat{\kappa_2} - \! c_{y_{j}}, z] \]

Suppose that 
 $\m[j, \widehat{\profit_1}, \widehat{\kappa_1}, \widehat{\profit_2}, \widehat{\kappa_2}, 1] =1$. Then, either 

 \begin{align}
\m[j\!-\!1, \widehat{\profit_1}\!-\!p_{x_{j}}, \widehat{\kappa_1}-c_{x_{j}}, \widehat{\profit_2}\!-\! p_{y_{j}}, \widehat{\kappa_2}\!-\!c_{y_{j}}, 1]=1 \label{first}      
 \end{align}
 or 
\begin{align}
\m[j\!-\!1, \widehat{\profit_1}\!-\!p_{x_{j}}, \widehat{\kappa_1}\!-\!c_{x_{j}}, \widehat{\profit_2}- p_{y_{j}}, \widehat{\kappa_2} \!-\! c_{y_{j}},2]=1 \label{second}
\end{align}

Since, the induction hypothesis applies to each of these entries, we will argue the case for \ref{first}, the argument for \ref{second} is symmetric. Suppose that \Cref{first} holds. Then, there exist $j$-length red and blue paths, denoted by $P_r$ and $P_b$ respectively, that are vertex disjoint such that $u(P_r) \geq \widehat{\profit_1}-p_{x_{j}}$ and $c(P_r) \leq \widehat{\cost_1}-c_{x_{j}}$ and $P_r$ ends at $x_{j}$. Moreover, $u(P_b) \geq \widehat{\profit_2}-p_{y_{j}}$ and $c(P_b) \leq \widehat{\cost_2}-c_{y_{j}}$. 

Hence, we note that the $(j+1)$-length paths $P_1=[P_r, x_{j}]$ and $P_2=[P_b, y_{j}]$ are red and blue paths that are vertex disjoint and $u(P_1) \geq \widehat{\profit_1}-p_{x_{j}} + p_{x_{j}}= \profit_1$ and $u(P_2) \geq \widehat{\profit_2}-p_{y_{j}} + p_{y_{j}} = \profit_2$; and $c(P_1) \leq \widehat{\cost_1}-c_{x_{j}} + c_{x_{j}}$ and $c(P_2) \leq \widehat{\cost_2}-c_{y_{j}} + c_{y_{j}}$. Moreover, $P_1$ ends at $x_{j}$ and $P_2$ at $y_{j}$. Hence, for this case the entry $(j, \widehat{\profit_1}, \widehat{\kappa_1}, \widehat{\profit_2}, \widehat{\kappa_2}, 1)$
is computed correctly. The argument for case given by ~\Cref{second} holds with symmetry. Therefore, we can conclude that the inductive step has been proved, and thus {\it all} entries are computed correctly. Consequently, $\lor_{z\in \{1,2\}}\m[r,P, B, P, B,z]$ gives the correct answer, as well. 

For the converse, we note that if $\lor_{z\in \{1,2\}}\m[r,P, B, P, B,z]=1$, then by backtracking as above, (given that each entry has been computed correctly) we can construct a solution for the instance explicitly. Clearly, this procedure can be executed in time $\BigO{r,(PB)^2}$ in addition to the polynomial processing required to create the bipartitions in each of the components. This concludes the proof of this theorem.
\end{proof}

\subsection{When the number of agents is arbitrary}

As discussed in the Introduction, interval graphs are a natural setting to study {\sc Job Scheduling}, a problem that is a special case of \pname. In this section we present a result, \Cref{cor:chordal}, which implies that when the \gname is an interval graph, \pname has pseudo-polynomial time algorithm. Specifically, we show that chordal graphs have this property and the result is due to a simple reduction to the problem when \gname has bounded treewidth, for which we show that \pname admits an \FPT algorithm with respect to $k+\tw$ and is pseudo-polynomial on $\alpha$ and $\gamma$.


\begin{theorem}\label{thm:ppoly-tw} 
\pname admits an $\FPT(k+\tw)$ algorithm when the values are encoded in unary. 

\hide{\il{old version}
When $k$ is a constant and the \gname has constant treewidth, \pname admits a pseudo-polynomial time algorithm. 

\il{Aliter: 1st version}
\pname admits an algorithm with running time $f(k, \tw){\rm poly}(P,B)$. 
}

\end{theorem}
\begin{proof}
Let \instance be an instance of \pname, where $k$ is constant and let $(T, \{\bt_{t}\}_{t \in V(T)})$ be a tree decomposition of \gname $G$ with treewidth $\tw$. We propose a dynamic programming based algorithm. Towards this end, we define the {\em PC-value} of a $k$-partition $(X_1, X_2, \ldots, X_k)$ of a set $S \sse V(G)$ to be a vector $((P_1,C_1), \ldots, (P_k,C_k))$ where for each $i \in [k]$,
$P_i = \sum\limits_{v \in X_i}p_v$ and $C_i = \sum\limits_{v \in X_i}c_v$. In our algorithm, we compute all possible PC-values that can be obtained from a $k$-independent partition of $G$. To design the dynamic programming algorithm, we consider the nice tree decomposition of $G$. Let $(T,\{\beta_t\}_{t \in V(T)})$ be a nice tree decomposition of $G$ of width $\tw$. Recall that, in nice tree decomposition, each node of the tree $T$ can be identified as a {\em leaf node, root node, introduce node, a forget node, or a join node}. 
\nc{For a node $t \in V(T)$, let $T_{t}$ denotes the subtree of $T$ rooted at $t$ and $V_{t} = \bigcup_{t \in V(T_t)}\bt_t$.} 
We traverse tree $T$  bottom-up and use a dynamic programming approach to compute a ``partial solution" for every node $t \in V(T)$ and every $k$-coloring $f$ of $G[\beta_t ]$.  For a node $t \in V(T)$ and a $k$-coloring $f$ of $G[\beta_t ]$,  the idea is to store the set of all PC-values, denoted by $\pc[t,f]$, that can be obtained from a $k$-coloring of $G[V_t]$ that agrees with $f$ on $\bt_{t}$.

Note that for the root $r$, bag $\beta_r$ is empty. Thus, every color class in a $k$-coloring of $\beta_r$ is empty, as well. Since $G[V_r] = G$ and any $k$-coloring of $G[V_r]$ agrees with the coloring $f$ of $\beta_r$, we can conclude that $\pc[r,f]$ contains all possible PC-values  which can be obtained from a $k$-coloring of $V(G)$.   Now, we compute $\pc[t,f]$ depending on the node type $t$.
\begin{enumerate}
    \item {\bf $t$ is the leaf (first) node:} \hide{Since we are considering a nice path decomposition,}Since $X_t = \emptyset$, each color  class in a $k$ coloring of $X_{t}$ is an empty set. Hence, \[\pc[t,f] = \{ ((0,0), (0,0), \ldots (0,0))\}.\]

    \smallskip
    \item {\bf $t$ is an introduce node:} By definition,
    $t$ has exactly one child, say $t'$, such that $\bt_{t} = \bt_{t'} \cup \{v\}$ for some vertex $v \in V \setminus \bt_{t'}$. 
    We compute $\pc[t,f]$ for a $k$-coloring $f$ of $\bt_{t}$ assuming that we have computed the values of $\pc[t',f']$ for all possible $k$-coloring $f'$ of $\bt_{t'}$. Let $f$ be a $k$-coloring of $\bt_{t}$ and $f' = f_{|_{\beta_t'}}$. 
    We define,
    \begin{align*}
    \label{pp_int_intro}
     \pc[t,f] = \{ \mathrm{b} + \Delta_i(p_v,c_v) \mid \!\mathrm{b} \in \pc[t',f'], i=f(v)\}, 
    \end{align*}
    where $\Delta_i(p_v,c_v)$ is a $k$-length vector whose $i$th entry is $(p_v,c_v)$ and all other entries are $(0,0)$. 


\hide{As I have already said the following should be part of the correctness argument.}
   \hide{ Now we show the correctness of the recurrence for introduce node. Consider a coloring $f$ of $\bt_{t}$ and let the introduced vertex $v$ is colored $i$, that is, $f(v)=i$. Clearly, $v$ is not adjacent to any vertex $u \in V_{t} \setminus \bt_{t}$. Thus any $k$ coloring $h$ of $V_t$ that agrees with the coloring $f$ of $\bt_{t}$, also agrees with the coloring $f'$ of $\bt_t$, where $f'$ is a restriction of $f$ on $\bt_{t'}$. Now, let $f'$ be a $k$-coloring of $\bt_{t'}$ such that $f'^{-1} \cap \bt_{t'} = \emptyset$. In this case, we note that any $k$-coloring $h$ of $V_t$ that agrees with $f'$ on $\bt_{t'}$ also agrees with the $k$ coloring $f$ of $\bt_{t}$ where $f(u)=f'(u)$ for all $u \in \bt_{t} \setminus \{v\}$ and $f(v) = i$. This proves a one-to-one relation between the set of $k$ coloring $h$ of $G[V_t]$ that agrees with $k$ coloring $f$ of $\bt_{t}$ and set of $k$ coloring $h$ of $G[V_{t'}]$ that agrees with $k$ coloring $f$ on $\bt_{t'}$, where $f'$ is a restriction of $f$. This proves the correctness of the recurrence. }
     
\item {\bf $t$ is a forget node:} $t$ has exactly one child $t'$, such that $\bt_{t} = \bt_{t'} \setminus \{v\}$ for some vertex $v \in \bt_{t'}$. We say a coloring $f$ of $\bt_{t}$ can be extendable to a coloring $f'$ of $\bt_{t'}$ if $f'(u) = f(u)$ for all $u \in \bt_{t}$ and $f'(v) = i$ if $f^{-1}(i) \cap N_{G}[v] = \emptyset$. For a forget node $t$, we define
\[
     \pc[t,f] = \bigcup_{f \text{ is extendable to } f' }\pc[t',f'] .
\]
\item{\bf $t$ is join node:} By definition of join node, $t$ has exactly two children $t_1$ and $t_2$ such that $\bt_{t} = \bt_{t_1} = \bt_{t_2}$. For a coloring $f$ of $\bt_{t}$, let $ (A_1,A_2, \ldots, A_k)= (f^{-1}(1), f^{-1}(2), \ldots, f^{-1}(k))$ and denote $\Tilde{b} = ((p(A_1), c(A_1)),(p(A_2), c(A_2)), \ldots, (p(A_k), c(A_k)))$. For a join node $t$, we define,
    \[
     \pc[t,f] = \{\mathrm{b}_{1} + \mathrm{b}_{2} - \Tilde{b} \mid \mathrm{b}_{1} \in \pc[t_1,f], \mathrm{b}_{2} \in \pc[t_2,f]\} .
    \]
\end{enumerate}

Finally, we solve the problem by checking if there is an element $((P_1,C_1), \ldots ,(P_k,C_k)) \in \pc[t,r]$ such that $P_i \ge P$ and $C_i \le B$ for each $i \in [k]$. If such an element exists in $\pc[t,r]$ then we conclude that $\mathcal{I}$ is a \yes-instance of \pname. The corresponding partition can be obtained by backtracking over the subproblems. 

\medskip
\noindent{\it Correctness.} We say that the entry $C[t,f]$ is computed {\it correctly} if it contains the set of all possible PC-values which can be obtained from a $k$-coloring of $G[V_t]$ that agrees with the $k$-coloring $f$ on $\beta_t$. We begin our analysis by observing that the correctness of the recurrence for a leaf node and the root follows trivially from the definition. Hence, our main analysis will focus on the introduce, forget, and join nodes. 

\smallskip
\noindent {\em When $t$ is a introduce node}: For an introduce node $t$, we show that $\pc[t,f] = \{ \mathrm{b} + \Delta_i(p_v,c_v) \mid \mathrm{b} \in \pc[t',f'] \}$. Note that $t$ has exactly one child $t'$ such that $\bt_{t} =  \bt_{t'} \cup \{v\}$, where $v \notin V_{t'}$. Consider a coloring $f$ of $\bt_{t}$ and let the introduced vertex $v$ be colored $i$, that is, $f(v)=i$. Clearly, 
any $k$ coloring $h$ of $V_t$ that agrees with the coloring $f$ of $\bt_{t}$, also agrees with the coloring $f'$ of $\bt_t$, where $f' = \restr{f}{\bt_{t'}}$ ($f'(v) = f(v)$ for each $v \in \bt_{t'}$).  
\begin{align*}
     \pc[t,f] \subseteq \{ \mathrm{b} + \Delta_i(p_v,c_v) \mid \!\mathrm{b} \in \pc[t',f'], i=f(v)\}.
\end{align*}
Now, let $f'$ be a $k$-coloring of $\bt_{t'}$ such that $f^{'-1}(i) \cap N_{G}(v) = \emptyset$. That is, no neighbor of $v$ is colored $i$ by $f'$. Also, using Theorem~\ref{thm:nice_tree_prop}, we note that $v$ is not adjacent to any vertex $u \in V_{t} \setminus \bt_{t}$. Therefore, in this case, any $k$-coloring $h$ of $V_{t'}$ that agrees with $f'$ on $\bt_{t'}$, also agrees with the $k$ coloring $f$ of $\bt_{t}$ where $f(u)=f'(u)$ for all $u \in \bt_{t} \setminus \{v\}$ and $f(v) = i$. 
\begin{align*}
      \{ \mathrm{b} + \Delta_i(p_v,c_v) \mid \!\mathrm{b} \in \pc[t',f'], i=f(v)\} \subseteq \pc[t,f]. 
\end{align*}
This proves that $C[t,f]$ has been computed correctly for an introduce node. 
Observe that, here we proved a one-to-one relation between the set of $k$ coloring $h$ of $G[V_t]$ that agrees with $k$ coloring $f$ of $\bt_{t}$ and set of $k$ coloring $h'$ of $G[V_{t'}]$ that agrees with $k$ coloring $f'$ on $\bt_{t'}$, where $f'= \restr{f}{\bt_{t'}}$. 

\smallskip
\noindent {\em When $t$ is a forget node}: Here $t$ has exactly one child $t'$ such that the corresponding bags satisfy $\bt_{t} = \bt_{t} \setminus \{v\}$ for some $v \in V$. Thus, $\bt_{t} \subseteq \bt_{t'}$ and $V_{t'} = V_{t}$. We say that a coloring $f$ of $\bt_{t}$ is  extendable to a coloring $f'$ of $\bt_{t'}$ if $f'(u) = f(u)$ for each $u \in \bt_{t}$ and $f'(v) =i$, where $i$ is a color class such that $f^{-1}(i) \cap N_{G}(v) = \emptyset$.

Now, consider an arbitrary $k$ coloring $h$ of $G[V_t]$ that agrees with a $k$ coloring $f$ of vertices in the bag $\bt_{t}$. Let $h(v) = i$ for some $i \in [k]$. It can be observed that the coloring $h$ agrees with all colorings $f'$ such that $f$ is extendable to $f'$. Thus we have,
    \[
     \pc[t,f] \subseteq \bigcup_{f \text{ is extendable to } f' }\pc[t',f'] .
    \]
Furthermore, if $h'$ is a coloring of $V_{t'}$ that agrees with a coloring $f'$ on $\bt_{t'}$ such that $f'(v) = i$ then $h'$ also agrees with the coloring $f$ on $\bt_{t}$ such that $f$ does not color any of the neighbour of $v$ with color $i$. Precisely, $h'$ also agrees with the coloring $f$ on $\bt_{t}$ such that $f^{-1}(i) \cap N_{G}(v) = \emptyset$. Therefore we have,
    \[
     \bigcup_{f \text{ is extendable to } f' }\pc[t',f'] \subseteq \pc[t,f].
    \]
This shows that we correctly compute all the subproblems at a forget node.

 

\smallskip
\noindent {\em When $t$ is a join node}: Recall that, by definition $t$ has exactly two children $t_1$ and $t_2$ such that $\bt_{t} = \bt_{t_1} = \bt_{t_2}$. Since $\bt_{t} = \bt_{t_1} = \bt_{t_2}$, $\bt_{t_1} \subseteq V_{t_1}$ and $\bt_{t_2} \subseteq V_{t_2}$, we have $\bt_{t} \subseteq V_{t_1} \cap V_{t_2}$. Also, using the property of tree decomposition, we note that no vertex of $V_{t_1} \setminus \bt_{t_1}$ is adjacent to a vertex $V_{t_2} \setminus \bt_{t_2}$. 

Let $h$ be a coloring of $V_t$ that agrees with a coloring $f$ on $\bt_t$. Then it is easy to see that the coloring $h_1$, a restriction of $h$ to $V_{t_1}$, agrees with $f$ on $\bt_{t_1}$. Similarly, the coloring $h_2$ that is exactly $h$ restricted to $V_{t_2}$, agrees with $f$ on $\bt_{t_2}$.

Furthermore, we also note that if $h_1$ is a coloring of $V_{t_1}$ that agrees with a coloring $f$ on $\bt_{t_1}$ and $h_2$ is a coloring of $V_{t_2}$ that  agrees with the coloring $f$ on $\bt_{t_2}$ then the coloring $h$ defined as, 

        \begin{equation*}
            h(v)= 
            \begin{cases} 
                 h_{1}(v),  & \text{ if } v \in V_{t_1} \setminus \beta_{t}, \\
                 h_{2}(v),  & \text{ otherwise},
            \end{cases}
        \end{equation*}
   
agrees with coloring $f$ on $\bt_{t}$. This is because the induced subgraphs $G[\bt_{t_1}]$, $G[\bt_{t_2}]$, and $G[\bt_{t}]$ are exactly same. Note that the vertices in $\bt_{t}$ are counted in the partition created due to coloring $h_1$ as well as coloring $h_2$. Thus the \utility and cost of vertices in $\bt_{t}$ are added twice. This concludes the correctness of the proposed recurrence for the join node. This concludes the proof of the theorem.


\medskip
\noindent{\it Time complexity:} Suppose the treewidth of the \gname $G$ in the input instance is $\tw$. Then there are at most~$\BigO{\tw \cdot n}$ nodes in nice tree decomposition. At each node $t$, we are considering all possible colorings of $\bt_{t}$. Thus, there are $k^{|\bt_{t}|}$ possible $k$-colorings of $\bt_{t}$. Furthermore, we observe that for any partition $(X_1,X_2, \ldots, X_{k})$ of $V_{t}$, we have, $(p(X_i), c(X_i)) \in \{0,1, \ldots, \alpha\} \times \{0,1, \ldots, \gamma\}$.  Therefore, each set $\pc[t,f]$ has at most $(\alpha \cdot \gamma)^{k}$ entries. The adjacency checking can be done in constant time. The computation time of $\pc[t,f]$ at each node depends on the type of node. For leaf nodes the $\pc[t,f]$ can be computed in constant time. At introduce node $t$, the $\pc[t,f]$ is computed using already computed values of corresponding $\pc[t',f']$ which takes time at most $(\alpha \cdot \gamma)^{k}$. For computing $\pc[t,f]$ when $t$ is a forget node, in $\BigO{k}$-time, we first check for all the color class which does not contain any neighbor of $v$. Then we can compute $\pc[t,f]$ in time $\BigO{k \cdot (\alpha \cdot \gamma)^{k}}$. At the join node for each entry of $\pc[t_1,f]$, we may have to go through all values in the set $\pc[t_2,f]$. Therefore, at join node $t$, $\pc[t,f]$ can be computed in time $\BigO{(\alpha \cdot \gamma)^{2k}}$. Hence, the total running time of the algorithm is $\BigO{\tw^{k+1} \cdot ( \alpha \cdot \gamma)^{2k} \cdot n}$, where $n$ denotes the number of vertices in the \gname.

This concludes that the \pname admits a pseudo-polynomial time algorithm when $k$ is constant and the \gname in the input instance has bounded treewidth. 
\end{proof}

Here are some nice observations that we can infer from \Cref{thm:ppoly-tw}. Let \instance be an instance of \gname, where $k= \BigO{1}$ and $G$ is a chordal graph.  Let $(T, \{\bt_{t}\}_{t \in V(T)})$ be a tree decomposition of $G$. It is known that a graph $G$ is a chordal graph if and only if it admits a tree decomposition where each bag induces a maximal clique in $G$, see \cite{chordal}. Such tree a decomposition of a chordal graph $G$ can be computed in linear-time~\cite{chordal_tw_linear}. Using these facts and Observation~\ref{obs:bounded-clique-size}, we note that if the treewidth of $G$ in $\mathcal{I}$ is more than $k-1$, then $\mathcal{I}$ is a \no-instance of \pname. Further, if the treewidth of $G$ is less than $k$, then we can use the algorithm given in Theorem~\ref{thm:ppoly-tw} to decide if $\mathcal{I}$ is a \yes-instance of \pname. Moreover, when $G$ is an interval graph (a subclass of chordal graph) the corresponding tree $T$ in the tree decomposition of $G$ is a path. Thus we do not have join node in case of interval graphs. From the discussion above, we have the following corollary.

\begin{cor}
\label{cor:chordal}
Let \instance be an instance of \pname. When $k$ is a constant and the \utility and the cost values are in unary,  we have the following:
\begin{enumerate}[label=(\alph*)]
\item if $G$ has constant treewidth, \pname admits a polynomial time algorithm;
\item if $G$ is chordal then \pname admits a polynomial time algorithm;
\end{enumerate}
\end{cor}

\begin{theorem*}
When $k$ is a constant and the \gname is an interval graph,
\pname admits a polynomial time algorithm with running time $\BigO{k^{k} \cdot ( \alpha \cdot \gamma)^{k} \cdot n}$.
\end{theorem*}
 

\begin{proof}
Let \instance be an instance of \pname where $G$ is an interval graph and $k$ is constant. To design our dynamic programming algorithm, we first compute a the nice path decomposition $(\mathcal{P}, (X_1,X_2, \ldots, X_r))$ of the \gname, $G$. If the pathwidth of $G$ is more than $k-1$, finding a required $k$ partition of $V(G)$ is not possible. Therefore $\mathcal{I}$ is a \no-instance of pathwidth of $G$ is more than $k-1$. Thus, we assume that the pathwidth of $G$ is at most $k-1$.  Specifically, we may assume that the pathwidth of $G$ is~$t \le k-1$.

Recall that, in nice path decomposition, each node of the path $\mathcal{P}$ can be identified as a {\em leaf node, root node, introduce node, or a forget node}. We traverse path $\mathcal{P}$  from $X_1$ to $X_r$ and use a dynamic programming approach to compute a ``partial solution" for every node $t \in V(\mathcal{P})$ and every $k$-coloring $f$ of $G[X_t]$. Here, we denote, $V_{t} = \bigcup_{s \le t} X_{s}$. The PC-values and subproblem for each node $t$ of $\mathcal{P}$ and a $k$-coloring $f$ of the vertices in the bag $\beta_{t}$ are same as we defined in Theorem~\ref{thm:ppoly-tw}. Recall the definition of subproblems,

\smallskip
\noindent
\fbox{
\parbox{0.45\textwidth}{
$\pc[t,f]$ -- set of all possible PC-values which can be obtained from a $k$-coloring $h$ of $G[V_t]$ that agree with $f$ on~$\beta_t$. 
}}

\medskip
By definition, $\beta_r$ is empty. Thus, every color class in a $k$-coloring of $\beta_r$ is empty, as well. Since $G[V_r] = G$ and any $k$ coloring of $G[V_r]$ agrees with the coloring $f$ of $\beta_r$, we can conclude that $\pc[r,f]$ contains all possible PC-values  which can be obtained from a $k$-coloring of $V(G)$. \hide{Note that each member of the set $\pc[r,f]$ is a $k$-tuple whose each entry is $2$-tuple; the first is the sum of \utility, and the second is the sum of costs.}  Now, we compute $\pc[t,f]$ depending on the node type $t$.

\begin{enumerate}
    \item {\bf $t_i$ is an introduce node:} By, definition, we have that $X_{t} = X_{t-1} \cup \{v\}$ for some vertex $v \notin X_{t-1}$. We compute $\pc[t_{i},f]$ for a $k$-coloring $f$ of $X_{t_{i}}$ assuming that we have computed the values of $\pc[t',f']$ for all possible $k$-coloring $f'$ of $X_{t'}$, where $t' \in [t-1]$. Let $f$ be a $k$-coloring of $X_{t_1}$ and $f'$ be the restriction of $f$ to the set $X_{t-1}$. 
    We define,
    \begin{align*}
    \label{pp_int_intro}
     \pc[t_i,f] \!=\! \{ \mathrm{b} + \Delta_i(p_v,c_v)\!\mid \!\mathrm{b}\!\in\!\pc[t_{i-1},f'], i\!=\!f(v)\}, 
    \end{align*}
    where $\Delta_i(p_v,c_v)$ is a $k$-length vector whose $i$th entry is $(p_v,c_v)$ and all other entries are $(0,0)$. 

    \smallskip
    \item {\bf $t$ is a forget node:} $t$ has exactly one child $t'$, such that $\bt_{t} = \bt_{t'} \setminus \{v\}$ for some vertex $v \in \bt_{t'}$. We say a coloring $f$ of $\bt_{t}$ can be extendable to a coloring $f'$ of $\bt_{t'}$ if $f'(u) = f(u)$ for all $u \in \bt_{t}$ and $f'(v) = i$ such that $f^{-1}(i) = \emptyset$. For a forget node $t$, we define
    \[
     \pc[t,f] = \bigcup_{f \text{ can be extended to } f' }\pc[t',f'] .
    \]

\end{enumerate}

We can give similar arguments as Theorem~\ref{thm:ppoly-tw} to show the correctness of the algorithm. Next, we analyze the time complexity.

\noindent{\it Time complexity.}  
The computation time of $\pc[t,f]$ remains same as we discussed in Theorem~\ref{thm:ppoly-tw} when $t$ is a leaf node, an introduce node, or a forget node. There is no join node in the path decomposition which save our time. When $G$ is an interval graph the total computation time of the algorithm is $\BigO{\pw^{k+1} \cdot ( \alpha \gamma)^{k} \cdot n} $
This concludes the proof of the Theorem.
\end{proof}



\paragraph{When the type is a constant.} Note that \Cref{rem:hardness} implies that \pname cannot have an $\FPT(k +\tw)$ or an $\FPT(k+\type)$ algorithm. Hence, it is worthwhile to consider $(k+\tw+\type)$.

\begin{theorem}
\label{thm:const_type}
When the type is a constant, \pname admits an $\FPT(k+\tw)$ algorithm.

\hide{
\il{Older version}
When the type and the treewidth of the \gname are constants, \pname is \FPT with respect to $k$.
\il{Or}
\pname admits an $\FPT(k+\tw + \type)$ algorithm.
}

\end{theorem}

\begin{proof}
Let \instance be an instance of \pname such that the treewidth of $G$ and the type are constants. 

\hide{Moreover, assume that there are only a constant number of different \utility and cost values on the items of $\Co{I}$. Let $A$ and $B$ be the set of different \utility and cost on items, respectively. That is, $A = \{p_v \mid v \in V(G)\}$ and $B = \{c_v \mid v\in V(G)\}$ where, $|A| = q$ and $|B|=r$. Note that, due to our assumption, $q$ and $r$ are constant numbers. Let $C = A \times B$. Note that $C =\{(p_v,c_v) \mid p_v \in A, c_v \in B\}$ consists of all possible \utility and cost pairs that items can have in $\Co{I}$.  and $|C| =\!q r\!=\!\size$.}

Fix an ordering $\Gamma = (r_1,r_2, \ldots, r_{\type})$ of the elements in $C=\{(p_v, c_v)\mid v\in V(G)\}$. We say that an item $u \in V(G)$ is of type $r_i$ if $(p_u,c_u) = r_i$. For a subset $U \subseteq V(G)$, we define the {\em configuration} of $U$, denoted by $\conf(U)$, to be an $\type$-length vector $(a_1,a_2, \ldots, a_{\type})$ where $a_i$ represents the number of items in $U$ that are of type $r_i$. Furthermore, for a $k$-independent partition of a subset $U \sse V(G)$, denoted by $\mathcal{X} = (X_1,X_2, \ldots, X_{k})$, we define the configuration of $\mathcal{X}$ as $\conf(\mathcal{X}) = (\conf(X_1), \conf(X_2), \ldots, \conf(X_k))$. Note that $\conf(\mathcal{X})$ is a $k$-length vector where the $i^{th}$ entry is an $s$-length vector representing the configuration of $X_i$. \hide{We use $0^i$ to denote a $i$-length vector whose each element is $0$.} Recall that a $k$-coloring of $V(G)$ is a $k$-independent partition of $V(G)$. In what follows, we will describe a dynamic programming-based algorithm to compute all possible configurations $\conf(\mathcal{X})$ for a $k$-coloring of $V(G)$.

Towards this end, we consider a nice tree decomposition \\$(T, \{\bt_{t}\}_{t \in V(T)})$ of $G$; and for a node $t \in V(T)$ and a $k$-coloring $f$ of $\beta_{t}$, we consider the subproblems defined by the entries $C[t,f]$. Intuitively speaking, these entries will store the set of all possible configurations for $k$-colorings of $G[V_t]$ that agree with $f$ on $G[\beta_t]$.

\smallskip
\noindent 
\fbox{
\parbox{0.45\textwidth}{
$\cpc[t,f]=$ set of all possible configurations for $k$-colorings of $G[V_t]$ that agree with $f$ on $G[\beta_t]$.
}}

Consequently, the set $C[r,f]$ contains all possible configurations for every $k$-coloring of the whole \gname $G$. Thus, a  
solution for the instance \Co{I}, if one exists. We traverse the tree $T$ in a bottom-up manner, and compute $\cpc[t,f]$ for all possible $k$-coloring $f$ of the vertices in $G[\bt_{t}]$. Depending on type of node $\cpc[t,f]$ can be computed as follows.

\begin{enumerate}
\item \textbf{$t$ is a leaf node:}  Since $\beta_t=\emptyset$, any $k$-coloring of $G[\bt_{t}]$ consists of empty sets. Thus, we have, 
\[\cpc[t,f]=\{(0^{\type},0^{\type},\ldots, 0^{\type})\} .\]

\item \textbf{$t$ is a introduce node:} Since, $t$ has exactly one child $t'$ such that $\beta_t=\beta_{t'}\cup \{v\}$, for some vertex $v \notin V_{t'}$, we have $V_t=V_{t'}\cup \{v\}$. Assume that $v$ is of type $r_j$ for some $j\in[ \type]$. 
Let $\delta_{j}$ denote the $\type$-length vector in which the $j^{th}$ entry is $1$ and all other entries are $0$. For a coloring $f$ of $G[\bt_{t}]$, we define the following recurrence,
\[
\cpc[t,f] = \{\mathrm{b} +\delta_{i,j} \mid \mathrm{b} \in \cpc[t',f'],
    i=f(v)\},
\]
        
where $\delta_{i,j}$ is a $k$-length vector whose $i^{th}$ entry is $\delta_{j}$ and all other entries are $0^{\type}$.
        
\item \textbf{$t$ is a forget node:} In this case, $t$ has exactly one child $t'$ such that $\beta_t=\beta_{t'}\setminus \{v\}$ for some vertex $v\in V\setminus \beta_t$.  Thus, $V_t=V_{t'}$. We say that a coloring $f$ of $G[\beta_t]$ is {\it extendable} to  the $k$-coloring $f'$ of $G[\beta_{t'}]$ if $\restr{f'}{\bt_t} = f$. Thus, for a forget node $t$, we define 
\begin{align*}
    \cpc[t,f] =  \bigcup_{f\text{ is extendable to } f'}\cpc[t',f']
\end{align*}

\item \textbf{$t$ is a join node:}  For a $k$-coloring $f$ of $G[\beta_t]$, we denote $\hat{b}_f = (\conf(f^{-1}(1)), \conf(f^{-1}(2)), \ldots, \conf(f^{-1}(k)))$. Since $t$ has two children $t_1$ and $t_2$ such that $\beta_t=\beta_{t_1}=\beta_{t_2}$, a $k$-coloring of $G[\beta_t]$ is also $k$-coloring of $G[\beta_{t_1}]$ and $G[\beta_{t_2}]$.  We define, 
\begin{align*}
    \cpc[t,f]=\{\mathrm{b} + \mathrm{b'}\!-\! \hat{b}_f\mid \mathrm{b} \in \cpc[t_1,f], \mathrm{b'} \in \cpc[t_2,f]\}.
\end{align*}
        
\item For $\cpc[r, f]$ we check if any of the configurations yields a feasible solution, we output ``yes", otherwise the answer is "no".
\end{enumerate}

\noindent  {\it Correctness. }  We prove the correctness of our approach as follows. We say that the entry $C[t,f]$ is {\it computed correctly} if it contains the set of all possible configurations for $k$-colorings of $G[V_t]$ that agree with $f$ on $G[\beta_t]$. We begin by observing that the correctness of the recurrence for a leaf node and the root follows trivially from the definition. Hence, our main analysis will focus on the introduce, forget, and join nodes. 

\smallskip
\noindent {\em When $t$ is a introduce node}: Note that $t$ has exactly one child $t'$ such that $\bt_{t} =  \bt_{t'} \cup \{v\}$, where $v \notin V_{t'}$. Consider a coloring $f$ of $\bt_{t}$ and assume that $f(v)=i$. Clearly,  any $k$ coloring $h$ of $V_t$ that agrees with the coloring $f$ of $\bt_{t}$, also agrees with the coloring $f'$ of $\bt_t$, where $f'$ is a restriction of $f$ on $\bt_{t'}$ ($f'(v) = f(v)$ for each $v \in \bt_{t'}$).  
\[
\cpc[t,f] \subseteq \{\mathrm{b} +\delta_{i,j} \mid \mathrm{b} \in \cpc[t',f'], i=f(v)\}.
\]
To show the reverse inclusion, let $f'$ be a $k$-coloring of $\bt_{t'}$ such that $f^{'-1}(i) \cap N_{G}(v) = \emptyset$. That is, no neighbor of $v$ is colored $i$ by $f'$. Also, due to Theorem~\ref{thm:nice_tree_prop}, $v$ is not adjacent to any vertex $u \in V_{t} \setminus \bt_{t}$. Thus, any $k$-coloring $h$ of $V_{t'}$ that agrees with $f'$ on $\bt_{t'}$, also agrees with the $k$ coloring $f$ of $\bt_{t}$ where $f(u)=f'(u)$ for all $u \in \bt_{t} \setminus \{v\}$ and $f(v) = i$. 
\[
\{\mathrm{b} +\delta_{i,j} \mid \mathrm{b} \in \cpc[t',f'],
                i=f(v)\}  \subseteq \cpc[t,f].
\]
This proves the correctness of computation of subproblems $\cpc[t,f]$ for each $k$-coloring $f$ of $\bt_{T}$ assuming that we have correctly computed $\cpc[t',f']$ for each $k$-coloring $f'$ of $\bt_{t'}$.  

\smallskip
\noindent {\em When $t$ is a forget node}: Here $t$ has exactly one child $t'$ such that the corresponding bags satisfy $\bt_{t} = \bt_{t} \setminus \{v\}$ for some $v \in V$. Thus, $\bt_{t} \subseteq \bt_{t'}$ and $V_{t'} = V_{t}$. We say that a coloring $f$ of $\bt_{t}$ is  {\em extendable to} a coloring $f'$ of $\bt_{t'}$ if $f'(u) = f(u)$ for each $u \in \bt_{t}$ and $f'(v) =i$, where $i$ is a color class such that $f^{-1}(i) \cap N_{G}(v) = \emptyset$.

Now, consider an arbitrary $k$ coloring $h$ of $G[V_t]$ that agrees with a $k$ coloring $f$ of vertices in the bag $\bt_{t}$. Let $h(v) = i$ for some $i \in [k]$. It can be observed that the coloring $h$ agrees with all colorings $f'$ such that $f$ is extendable to $f'$. Thus we have,
    \[
     \cpc[t,f] \subseteq \bigcup_{f \text{ is extendable to } f' }\cpc[t',f'] .
    \]
Furthermore, if $h'$ is a coloring of $V_{t'}$ that agrees with a coloring $f'$ on $\bt_{t'}$ such that $f'(v) = i$ then $h'$ also agrees with the coloring $f$ on $\bt_{t}$ such that $f$ does not color any of the neighbour of $v$ with color $i$. Precisely, $h'$ also agrees with the coloring $f$ on $\bt_{t}$ such that $f^{-1}(i) \cap N_{G}(v) = \emptyset$. Therefore we have,
    \[
     \bigcup_{f \text{ is extendable to } f' }\cpc[t',f'] \subseteq \cpc[t,f].
    \]
This shows that we correctly compute all the subproblems at a forget node.

 

\smallskip
\noindent {\em When $t$ is a join node}: Recall that, by definition $t$ has exactly two children $t_1$ and $t_2$ such that $\bt_{t} = \bt_{t_1} = \bt_{t_2}$. Since $\bt_{t} = \bt_{t_1} = \bt_{t_2}$, $\bt_{t_1} \subseteq V_{t_1}$ and $\bt_{t_2} \subseteq V_{t_2}$, we have $\bt_{t} \subseteq V_{t_1} \cap V_{t_2}$. Also, using the property of tree decomposition, we note that no vertex of $V_{t_1} \setminus \bt_{t_1}$ is adjacent to a vertex $V_{t_2} \setminus \bt_{t_2}$. 

Let $h$ be a coloring of $V_t$ that agrees with a coloring $f$ on $\bt_t$. Then it is easy to see that the coloring $h_1$, a restriction of $h$ to $V_{t_1}$, agrees with $f$ on $\bt_{t_1}$. Similarly, the coloring $h_2$ that is exactly $h$ restricted to $V_{t_2}$, agrees with $f$ on $\bt_{t_2}$. 

Furthermore, we also note that if $h_1$ is a coloring of $V_{t_1}$ that agrees with a coloring $f$ on $\bt_{t_1}$ and $h_2$ is a coloring of $V_{t_2}$ that  agrees with the coloring $f$ on $\bt_{t_2}$ then the coloring $h$ defined as, 
        \begin{equation*}
            h(v)= 
            \begin{cases} 
                 h_{1}(v),  & \text{ if } v \in V_{t_1} \setminus \bt_{t},\\
                 h_{2}(v),  & \text{ otherwise},
            \end{cases}
        \end{equation*}
agrees with coloring $f$ on $\bt_{t}$. This is because the induced subgraphs $G[\bt_{t_1}]$, $G[\bt_{t_2}]$, and $G[\bt_{t}]$ are exactly same. Note that the vertices in $\bt_{t}$ are included in the partition created due to coloring $h_1$ as well as coloring $h_2$. Thus the configuration of partition of $\bt_{t}$ created by $f$ are added twice. This concludes the correctness of the proposed recurrence for the join node. This concludes the proof of the theorem.

\smallskip
\noindent \emph{Time complexity:}  At each node $t$, we maintain all possible configurations of a $k$-independent partition of $V_t$. Note that for any set $X \subseteq V$, $\conf(X) \in \{0,1,\ldots, s\}$. Thus for any $k$-coloring $f$ of $\bt_{t}$, the set $\cpc[t,f]$ has at most $(s+1)^{k}$ elements. At introduce node $t$, the $\cpc[t,f]$ is computed using already computed values of corresponding $\cpc[t',f']$ which takes time at most $s^{k}$. For computing $\cpc[t,f]$ when $t$ is a forget node, in $\BigO{k}$-time, we first check for all the color class which does not contain any neighbor of $v$. Then we can compute $\cpc[t,f]$ in time $\BigO{k \cdot \type^{k}}$. At the join node for each entry of $\cpc[t_1,f]$, we may have to go through all values in the set $\cpc[t_2,f]$. Therefore, at join node $t$, $\cpc[t,f]$ can be computed in time $\BigO{\type^{2k}}$. This concludes that, in time $\BigO{\tw^k \cdot (\type+1)^{2k} \cdot k \cdot n}$ we can decide if $\mathcal{I}$ is a \yes-instance of \pname. If $\mathcal{I}$ is a \yes-instance a corresponding $k$-independent partition be obtained by backtracking.  
\end{proof}

Let \instance be an instance of \pname. From Theorem~\ref{thm:const_type}, we infer that, if $k$ and treewidth of \gname in $\mathcal{I}$, are constant then \pname can be solved in polynomial-time. As we discussed, if $G$ is a chordal graph and $\tw(G) \ge k$ then $\mathcal{I}$ is a \no-instance of \pname. Furthermore if $\tw(G) \le k-1$ then for constant $k$, \pname can be solved in polynomial-time. Thus we have the following corollaries.

\begin{cor}
\label{cor:constant_type}
When type is a constant and the \gname is chordal, \pname
\begin{enumerate}[wide=0pt]
    \item[(a)] is \FPT with respect to $k$; and 
    \item[(b)]  when $k$ is also a constant it can be solved in $\BigO{n}$ time.
\end{enumerate}
\end{cor}





\paragraph{Exact Exponential Algorithm.} We will present an exact exponential algorithm that runs in time $\Co{O}^{\star}(2^{n})$. It is based on the principle of subset convolution and is implemented via polynomial algebra, specifically FFT based polynomial multiplication, which guarantees that we can multiply two polynomials of degree $d$ in at most \BigO{d\log d} time,\cite{moenck1976practical}. The basic idea is that we build our solution, $(S_1, \ldots, S_k)$, where each $S_i$ is a piece, round by round, with each round keeping track of the subsets that could give rise to a piece in the (hypothetical) solution of \pname. More specifically, in the first round we define monomials that correspond to subsets of items that satisfy the \utility and cost constraint and therefore could be the first piece in the solution; in the second round we define a polynomial in which each monomial represents a solution of the subproblem defined for two agents, and so on. After $k$ rounds, we have a polynomial which is not identically zero if and only if there is a solution for \pname. 



\begin{theorem}\label{thm:exact-algorithm}For any $k\geq 1$,
\pname can be solved in \nc{$\mathcal{O}^{\star}(2^{n})$}-time.
\end{theorem}

\begin{proof}
We present an algorithm for \pname using the subset convolution technique in the specified time.  Our objective is to compute a $k$-partition of the vertices of $G$ such that each set $S$ in the partition satisfies the following properties; (i) $S$ is an independent set, (ii) $c(S)\leq B$, and (iii) $p(S)\geq P$.  We define the following indicator function $f$ such that $f: 2^{V}\rightarrow \{0,1\}$.  For a set $S\in 2^{V}$, $f(S)=1$ if $S$ is an independent set, $c(S)\leq B$, and $p(S)\geq P$, and 0, otherwise.  
We use polynomial multiplication technique to find such a $k$-partition of $V$.  Let $V=\{v_1,\ldots,v_n\}$.  For a subset $S\subseteq V$, a characteristic vector $\chi(S)$ is defined as follows, $\chi(S)[i] = 1$, if $u_i\in S$, and $0$ otherwise.

Observe that for a set $S\subseteq V$, such that $\chi(S)$ is an $n$-length binary vector.  Two strings $S_1$, and $S_2$ are said to be disjoint if $\forall i\in [n]$, $\chi(S_1)[i]\neq \chi(S_2)[i]$.  The {\it Hamming weight} of a binary string is the number of ones in the binary representation.  For a binary string $S$, let $\hm(S)$ denote the Hamming weight of $S$.  The Hamming weight of a monomial $x^i$ is the Hamming weight of $i$ (expressed as a binary vector).  The following lemma captures the relationship between disjointness of sets and Hamming weights.

\begin{lemma}
{\rm \cite{CYGAN20103701}}\label{lem:Ham-weight}
Subsets $S_1,S_2\subseteq V$ are disjoint if and only if Hamming weight of $\chi(S_1)+\chi(S_2)$ is $|S_1|+|S_2|$.
\end{lemma}

We use $\hm_l(F(X))$ to denote the {\it Hamming projection} of a polynomial $F(x)$ to $s$ which is the sum of all monomials in $F(X)$ having Hamming weight $s$.  We use $\rep(F(X))$ to denote the {\it representative polynomial} of $F(X)$, the coefficient of the monomial is one if the coefficient of the monomial is non-zero.  
Given an instance \instance, our objective is to find a polynomial which is non-zero if and only if $\mathcal{I}$ is a \yes-instance.  Our algorithm finds the polynomial iteratively. We define a polynomial $F^1_{s}(x)$ of type 1 as follows for each $s\in [n]$, $$F^1_s(x)=\sum_{Y\subseteq U, |Y|=s}f(Y)x^{\chi (Y)}$$
    Observe that the polynomial $F^1_s(x)$ contains all subsets of fixed size $s$ for which the indicator function $f$ is satisfied.  Further, all the polynomials in $F^1(x)=\{F^1_s(x)\mid s\in [n]\}$ together contain information about all possible subsets of $V$ that are independent and satisfy the \utility and cost constraints.  For each $s\in [n]$, we define the polynomials of type $j\in[k]\setminus [1]$ as follows
    $$F^j_s(x)=\sum_{\substack{ s_1,s_2\in [n]: s_1+s_2=s}}\rep(\hm_{s}(F^1_{s_1}\times F^{j-1}_{s_2}))$$
    
Thus, a polynomial of type $j$ is obtained by multiplying polynomial of type 1 with polynomial of type $j-1$. Observe that $\hm(\cdot)$ function ensures that the subsets  corresponding to each monomial are formed by disjoint union of smaller set, and $\rep(\cdot)$ function ensures that coefficients of all non-zero monomials is one. The solution to $\mathcal{I}$ reduces to checking whether $F^k_n(x)$ is a non-zero polynomial.  If $F^k_n(x)$ is a non-zero polynomial, then the given instance is a \yes-instance; otherwise a \no-instance. 

Next, we  analyze the proof of correctness of our algorithm by using the following lemma.
   
\begin{lemma}
$F^k_n(x)$ is a non-zero polynomial if and only if $\mathcal{I}$ is a \yes-instance.
\end{lemma}
\begin{proof}($\Rightarrow$) We will prove that for $j\in [k]$, if polynomial $F^j_s(x)$ contains a non-zero monomial, for any $s\in [n]$, then there exist subsets $Y_1,\ldots,Y_j$ of $V$ such that they are pairwise disjoint, $f(Y_i)=1$, for each $i\in [j]$ and $|Y_1|+\ldots+|Y_j|=s$.  That is, each $Y_1,\ldots,Y_j$ satisfies the constraints of the indicator function: the subsets $V_1,\ldots,V_j$are independent, and that each $Y_i$ respects \utility and cost constraints. Our proof is by induction on $j$.
        
\emph{Base case:}  For $j=1$, if for some $s\in [n]$ such that $F^1_s(x)$ is non-zero, then it has a non-zero monomial, say $x^{(\chi(Y_1))}$ such that $|Y_1|=s$ and $f(Y_1)=1$.  This implies that there exist a set $Y_1\subseteq V$ which is independent which respects the \utility and cost constraints.

\emph{Induction step:}  We will assume that the statement is true for some $j<k$ and we shall prove for $j+1$.  Observe that the polynomial $F^{j+1}_s(x)$ is obtained from a pair of polynomials $F^1_{s_1}$ and $F^j_{s_2}$ where $s_1+s_2=s$.  If $F^{j+1}_s(x)$ has a non-zero monomial, say $x^{\chi(S)}$, for some $S$, then there exists a non-zero monomial in both $F^1_{s_1}(x)$ and $F^j_{s_2}(x)$ such that $s_1+s_2=\hm(S)$.  By induction hypothesis, the following statements are true.
\\(i) There exists subsets $Y_1,\ldots,Y_j$ such that they are pairwise disjoint, $|Y_1|+\ldots+|Y_j|=s_2$ and each set satisfies the indicator function.
\\(ii) There exist a monomial $x^{Y'}$ in $F^1_{s_1}(x)$ such that $|Y'|=s_1$ and $Y'$ satisfies the indicator function.
\\By our construction and by \cref{lem:Ham-weight}, $\hm(\cdot)$ function on the polynomial $F^{j+1}_s(x)$ implies that $Y_1\cup\ldots\cup Y_j$ and $ Y'$ are pairwise disjoint.  Hence, we have $j+1$ pairwise disjoint subsets of $V$ satisfying the indicator function.  Therefore, we can conclude that if $F^k_n(x)$ is non-zero, then $\mathcal{I}$ is a \yes-instance.

($\Leftarrow$) We will prove that if $\mathcal{I}$ is a \yes-instance, then $F^k_n(x)$ is a non-zero polynomial.  Let $S_1,\ldots,S_k$ be the partition of $V$ respecting the \utility and cost constraint.  We shall prove that there exists a non-zero monomial $x^{\chi(S_1,\ldots,S_k)}$.  In order to prove that there exists a non-zero monomial in $F^k_n(x)$, we shall prove that each $j\in [k]$, $F^j_n(x)$ is a non-zero polynomial.  The proof is by induction on $j$.  
        
\emph{Base case:}  For $j=1$, recall that $F^1(x)=\{F^1_s(x)\mid s\in [n]\}$.  Since the subsets $S_1,\ldots,S_k$ satisfies constraints of the indicator function and they are partition of $V$, $\forall i\in [n]$, $x^{\chi(S_i)}$ is a non-zero monomial in some polynomial precisely $F^1_{|S_i|}(x)$.  
        
\emph{Induction step:} We will assume that the claim is true for some $j<k$ and we shall prove for $j+1$.  By induction hypothesis, we know that there is a polynomial in $F^j(x)$ which has a non-zero monomial $x^{\chi(S_1\cup S_2\cup \ldots\cup S_j)}$.  Further, by base case we know that there exists a polynomial in $F^1(x)$ which has a non-zero monomial $x^{\chi(S_{j+1})}$.  Since $S_1,\ldots,S_{j+1}$ are pairwise disjoint subsets of $V$, $F^{j+1}_s(x)$ where $s=\sum_{i\in [j]}|S_i|$ contains a non-zero polynomial such that it has a non-zero monomial $x^{\chi(S_1\cup \ldots \cup S_{j+1})}$.

This concludes the proof of the lemma.
\end{proof}

\begin{proposition}{\rm \cite{moenck1976practical}}
\label{prop:poly-mult}The product of two univariate polynomial of degree $d$ can be computed in time \BigO{d\log d}.
\end{proposition}

\noindent\emph{Time complexity:} 
Any polynomial generated for a subset $S\subseteq V$ has degree at most $2^n$.  The product of two polynomials can be computed in time $\mathcal{O}(2^nn)$ by \cref{prop:poly-mult}.  The total number of polynomials is at most $nk$.  Thus the running time of our algorithm is $\Co{O}^{\star}(2^{n})$.
\end{proof}

We know from the following influential result that $k$-{\sc Coloring}  has a lower bound unless ETH fails. 

\begin{proposition}
{\rm \cite{Lower-bound-k-coloring}}
Unless {\em ETH} fails, $k$-{\sc Coloring} cannot be solved in $2^{o(n)}$ time. 
\end{proposition}
Thus, \Cref{thm:exact-algorithm} is asymptotically tight.

\paragraph{Parameterized Algorithm}





 Due to \Cref{rem:hardness}, we cannot expect an \FPT algorithms with respect to parameters $k+P+B$, or $s$, where $s$ denotes the maximum bundle size. However, as shown in the next result, we can expect one with respect to $k+s$. This is because $ks\geq n$ and so an exhaustive search yields an $\FPT(k+s)$ algorithm. 


\hide{Hence, we look at $s$, the bundle size, \hide{defined as the maximum size of a set assigned to an agent in a solution of \pname,} as a parameter. The bundle size is likely to be small even if the profit and/or cost of the bundle is not. This parameter is in general unrelated to $k$ or the other numerical values in instance besides $P \geq s$. But due to this condition alone an \FPT algorithm with respect to $s$ will be an algorithm that runs in time $f(P)n^{\BigO{1}}$, for some function $f$. This would yield a polynomial-time algorithm when $P$ is polynomially bounded in $n$, and thereby  contradict the strong NP-hardness of \pname. As explained earlier in the Introduction, we cannot expect to have an algorithm with running time $f(P)n^{\BigO{1}}$ and since $P\geq s$, we cannot have an \FPT algorithm with respect to $s$. However, as shown in the next result, we can expect one with respect to $k+s$. This is because $ks\leq n$ and so an exhaustive search yields an $\FPT(k+s)$ algorithm.}

\begin{theorem}
    \label{thm:agent-and-bundle-size}
    \pname admits an $\FPT(k+s)$ algorithm.
    


\end{theorem}


For the special case when a bundle size is at most two, we can reduce the question to finding a matching in an auxiliary graph, where each edge represents a feasible bundle that satisfies the profit and cost constraints. Hence, the problem is polynomial-time solvable, in contrast to the hardness for the case $s=3$, \Cref{rem:hardness}.
\begin{theorem}\label{thm:bundlesize-2}
When $s=2$, \pname admits a polynomial-time algorithm.
\end{theorem}

\begin{proof}
Let \instance be an instance of \pname where $s =2$. We will prove the result when every bundle is of size two.
If the size of the vertex set of $G$ is not $2k$, then we output ``no''. Else, we define an auxillary graph $H$ on the vertex set $V(G)$, such that $\{a, b\}\sse V(G)\sm E(G)$ and $u(a) + u(b) \geq P$ and $c(a) + c(b) \leq B$. That is, every feasible 2-independent set of size two among the vertices of $G$ share an edge in $H$. 
We compute the maximum matching in $H$. If it is not perfect, we answer ``no''. Else, we  answer ``yes" and output the $k$ bundles which correspond to the matching edges in $H$. 
The proof of correctness is quite easy to see, but we will nevertheless, discuss for the sake of completeness.

Suppose that \Co{I} is a \yes-instance, where every bundle has size two. Then, clearly $|V(G)|=2k$ and there is a perfect matching in $H$. Hence, our algorithm will be able to detect it. If our algorithm outputs ``yes'', then the matching in $H$ constitutes $2$-independent partition in $G$. Since, $|V(G)|=2k$, so we can conclude that \Co{I} is a \yes-instance. 

\end{proof}

\paragraph{Approximation Algorithm.} Next, we present an \FPTAS for \pname by repurposing the dynamic programming approach of~\Cref{thm:ppoly-tw}
and the standard bucketing technique to split the range of the utility and cost values into smaller ``buckets" so as to approximately estimate the value of each subproblem. Specifically, we show that for any positive real numbers $\epsilon, \omega\!>\!0$ and a $k$-independent partition $(X_1,  \ldots, X_k)$ with PC-value $((P_1,C_1), \ldots, (P_k, C_k))$ obtained in \Cref{thm:ppoly-tw}, we can obtain an approximate PC-value $((\wh{P_1},\wh{C_1}), \ldots, (\wh{P_k},\wh{C_k}))$ such that $P_j \le (1 + \epsilon)\wh{P_{i}}$ and $\wh{C_j} \le (1 + \omega)C_{i}$ for each $j \in [k]$.


\begin{theorem}\label{thm:FPTAS-tw}
\pname admits an algorithm that in $\FPT(k+\tw)$ time outputs a solution in which every agent has profit at least $P/(1+\epsilon)$
 and cost at least $(1+\omega)B$.





\end{theorem}
\begin{proof}
Let \instance be an instance of \pname, where $k$ is a constant. Recall that $\alpha$ and $\gamma$ are the aggregate \utility and cost of all the items, respectively. For a constant $k$, the pseudo-polynomial algorithm in Theorem~\ref{thm:ppoly-tw} works as follows. At a node $t$ and for a coloring $f$ of $\bt_{t}$, the defined subproblem $\pc[t,f]$ is the set of all possible PC-values of every $k$-coloring of $G[V_t]$ that agrees with $f$ on $\bt_{t}$. Thus, at the root node $r$, we have all possible PC-values of any $k$-coloring of $G$. Then we check which all PC-values satisfy the \utility and cost constraints. If none of them meet the constraints, then $\mathcal{I}$ is a \no-instance. 

Indeed, at a node $t$, we may be updating $(\profit \cost)^{k}$ (at most) PC-values for coloring $f$ of $G[\bt_{t}]$ which is a primary challenge in achieving a polynomial-time algorithm. Consequently, at a node $t$ and coloring $f$ of $\bt_t$, we store approximate PC-values of a $k$-independent partition of $G[V_{t}]$ that agrees with $f$ on $\bt_{t}$. Towards the end, for positive real numbers $\epsilon, \omega\!>\!0$ and for any $k$-independent partition $(X_1,  \ldots, X_k)$ with PC-value $((P_1,C_1), \ldots, (P_k, C_k))$, we want to obtain a {\it relaxed} PC-value $((\wh{P_1},\wh{C_1}), \ldots, (\wh{P_k},\wh{C_k}))$, such that $P_j \le (1 + \epsilon)\wh{P_{i}}$ and $\wh{C_j} \le (1 + \omega)C_{i}$ for each $j \in [k]$. 

To achieve this, we divide the range of \utility and cost value. Note that, in $k$-independent partition $(X_1,X_2, \ldots, X_{k})$ of $V(G)$, the \utility and the cost of a set $X_i$ ranges from $0$ to $\profit$ and $0$ to $\cost$, respectively. For a given positive value $0<\epsilon \leq 1$, we split the \utility value of each set $X_i$ into intervals whose endpoints vary by a factor of $(1+\epsilon)$: 
\begin{equation*}
    \begin{split}
    (0, (1+\epsilon)^{\frac{1}{n}}],\ldots, ((1+\epsilon)^{\frac{i-1}{n}}, (1+\epsilon)^{\frac{i}{n}}], \ldots 
    \end{split}
\end{equation*}

where, $i\leq \ceil{n\log_{1+\epsilon} \alpha}$.
Similarly, for $\omega \in (0,1)$, we split the cost value of each set $X_i$ into $s = \ceil{ n \cdot \log_{(1+ \omega)}\gamma}$ intervals whose end points vary by a factor of $(1+\omega)$:
\begin{equation*}
    \begin{split}
    (0,(1+\omega)^{\frac{1}{n}}], \ldots, ((1+\omega)^{\frac{i-1}{n}}),(1+\omega)^{\frac{i}{n}})],\ldots 
    \end{split}
\end{equation*}

Note that at an introduce node $t$, we look for the possibility of assigning vertex $v$ to some set in a $k$-partition $(X_1, X_2, \ldots, X_k)$ of $V_{t}$ and the $PC$-value corresponding to this partition gets updated. Let $((P_1,C_1),\ldots, (P_k,C_k))$ denote the PC value before assigning $v$ to the set $X_j$. Thus, after update, the resulting PC-value is $(\ldots, (P_j+p_{v}, C_j+c_v),\ldots, (P_k,C_k))$. Similar arguments hold for the other nodes as well. For every operation, we can pin-point the sets in the PC value that is affected. This allows us to bound the change in value of the PC-values in an approximate manner.

Specifically, for the introduce node, the updated \utility of $X_j$ is rounded down to the nearest interval endpoint to which the value below belongs. That is, suppose that $P_j+ p_v \in [(1+\epsilon)^{\frac{i}{n}}, (1+\epsilon)^{\frac{i+1}{n}})$ for some $i \leq \ceil{n \log_{(1+\epsilon}\profit}$, then the approximate \utility value is at least $(1+\epsilon)^{\frac{i}{n}}$. Similarly, the updated cost value is rounded up to the nearest interval endpoint. Hence, $C_j+c_v$ is at most $(1+\omega)^{\frac{i'}{n}}$ for some $i' \leq \ceil{n\log_{(1+\omega)}\gamma}$. Similar argument hold for the forget and join nodes as well.The \utility (and cost) values of the sets that are affected by our operation can rounded down (and up), as described above. Through the next claim, we establish a relationship between the standard PC-value computed by the exact algorithm (of \Cref{thm:ppoly-tw}) and the relaxed PC-value computed by the approximation algorithm at every step. The factor of $(1+\epsilon, 1+\omega)$ follows as a result.
\begin{claim}
\label{claim:fptas}
Consider a PC-value, $((P_1, C_1),  \ldots,  (P_k, C_k))$, generated by the algorithm \Cref{thm:ppoly-tw} after assigning the first $i-1$ vertices. Suppose that the $i^{th}$ vertex affects the $j^{th}$ set. Then, there exists a relaxed PC-value $((\wh{P_1},\wh{C_1}), \ldots, (\wh{P_k},\wh{C_k}))$ such that $P_j \le (1 + \epsilon)^{i/n}\wh{P_{j}}$ and $\wh{C_j} \le (1 +\omega)^{i/n}C_{j}$. 
\end{claim}
\begin{proof}
   We prove this using induction on $i$. For $i=1$, there is only one vertex $v$ in the $j^{th}$ set of the partition. Thus $P_j = p_v$ and $C_{j} = c_v$. In the FPTAS, there is a state where $P'_j$ is the largest lower interval endpoint not exceeding $P_j$, and $C'_j$ is the lowest above interval endpoint exceeding $C_j$. It follows from the construction of the intervals that $(1 + \epsilon)^{1/n}P'_{j} \ge P_{j}$ and $(1 + \omega)^{1/n}C_{j} \ge C'_{j}$. 

   Suppose the claim is true for first $i-1$ vertices and now we are considering the assignment of $i^{th}$ vertex, say $v$.  Suppose that the $i^{th}$ vertex affects the $j^{th}$ set. That is, $i^{th}$ vertex is assigned to the $j^{th}$ set of $k$-independent partition $(X_1,\ldots, X_k)$. Precisely, the \utility and cost of $X_j$ is updated to $P_j +p_v$ and $C_j + c_v$, respectively. Due to induction hypothesis, before this update the $P_j$ and $C_j$ there exists a relaxed value $(\wh{P_j}, \wh{C_j})$ that satisfy $P_j \le (1 + \epsilon)^{(i-1)/n}\wh{P_{j}}$ and $\wh{C_j} \le (1 + \omega)^{(i-1)/n}C_{j}$. While updating the \utility, $p_v$ is added to $\wh{P_j}$, and the result is rounded down to a lower \utility interval endpoint $\wh{P_j'}$ such that $(1+ \epsilon)^{1/n}\wh{P_j'} \ge \wh{P_j} +p_v$. Using induction hypothesis we have that  $(1+ \epsilon)^{1/n}\wh{P_j'} \ge \wh{P_j} +p_v \ge (1 + \epsilon)^{-((i-1)/n)}P_j + p_v \ge (1 + \epsilon)^{-((i-1)/n)}(P_j + p_v)$. This concludes that, $P_j +p_v \le (1+ \epsilon)^{i/n}\wh{P_j'}$. In similar manner, while updating the cost, $c_v$ is added to $\wh{C_j}$, and the result is rounded down to a upper cost interval endpoint $\wh{C_j'}$ such that $(1+ \omega)^{1/n}(\wh{C_j} +c_v ) \ge \wh{C'_j} $. Using induction hypothesis we have that  $\wh{C'_j} \le (1+ \omega)^{1/n}( (1+ \omega)^{i-1}{C_j} +c_v )) \le (1+ \omega)^{1/n}((1+ \omega)^{i-1}({C_j} +c_v )) \le (1 + \omega)^{i/n}(C_j + v_v)$. Therefore the claim follows. 
\end{proof}

Note that each set $\pc[t,f]$ in the approximation algorithm is consists of at most $r \cdot s$ entries. Thus the running time of the algorithm is $\BigO{\tw^{k} \cdot (rs)^{2k} \cdot n}$, where where, $r= \ceil{n\log_{1+\epsilon} \alpha}$ and $s = \ceil{ n \cdot \log_{(1+ \omega)}\gamma}$. Hence, when $k$ is a constant, the result follows from Theorem~\ref{thm:ppoly-tw} and Claim~\ref{claim:fptas}. 
\end{proof}
Consequently, \Cref{thm:FPTAS-tw} yields the following.

\begin{cor}
    When $k$ is a constant and the \gname has constant treewidth, \pname admits an \FPTAS.

\end{cor}

\subsection{About general utility and cost functions}
\label{ss:generalizing-to-gen-valuations}
Throughout the paper, we have presented algorithms for identical \utility and cost functions (between agents). These results can be extended to arbitrary additive \utility and cost functions. Consider the polynomial-time algorithm in \Cref{thm:ppoly-tw}. In case of arbitrary additive \utility and cost functions we replace $\alpha$ with $\wh{\alpha}=\max_{i\in[k]}\max\{p_i(j)|j\in [n]\}$ and $\gamma$ with $\widehat{\gamma}=\max_{i\in[k]}\max\{c_i(j)|j\in [n]\}$. Moreover, while updating the PC-values of a $k$-independent partition $(X_1,X_2, \ldots, X_k)$ of $V_{t}$ at any node $t$, to update the \utility and cost of $X_i$ we use the valuation and cost function of agent $j$. Similarly, in \Cref{thm:const_type}, \hide{instead of defining the type of an instance as set of different utility and cost pairs,} we can redefine the {\em type} of an instance as the union of set of different utility and cost pairs of each agent. Then, the type of elements in $X_j$ can be updated according to the functions of agent $j$. Thus, the results of \Cref{thm:ppoly-tw,thm:const_type} and their corollaries can be extended to general valuation and cost functions. 


\section*{Future Direction}
We initiated the study of conflict-free fair division under budget constraints. This has applications to various real-life scenarios. We considered the egalitarian fair division as the fairness criteria for our study. Below we describe a few future research directions that we view as important and promising: (i) other fairness notions, such as envy-freeness, NSW, maximin, Pareto optimality, etc, (ii) dependencies between the items, e.g connectivity (iii) two-sided preferences.


\vspace{0.3cm}
\noindent\textbf{Acknowledgement:} Sushmita Gupta, Pallavi Jain, and Vikash Tripathi acknowledges their support from SERB-SUPRA grant number S/SERB/PJ/20220047.  Sushmita Gupta also acknowledges support from SERB's MATRICS grant (MTR/2021/000869).  Pallavi Jain also acknowledges support from IITJ Seed Grant grant I/SEED/PJ/20210119.

\bibliographystyle{ACM-Reference-Format} 
\bibliography{arXiv_references.bib}
\clearpage






\end{document}